\documentclass[10pt]{article}

\usepackage[T1]{fontenc}
\usepackage[letterpaper,textwidth=6.8in,textheight=9in,centering]{geometry}
\usepackage[parfill]{parskip}
\usepackage{amsmath,amssymb,amsthm,mathtools}
\usepackage[tt=false,type1=true]{libertine}
\usepackage[varqu]{zi4}
\usepackage[libertine]{newtxmath}
\usepackage{microtype,booktabs,tabularx,enumitem,needspace}
\usepackage{graphicx,tikz}
\usetikzlibrary{arrows.meta,positioning}
\usepackage{algorithm,algpseudocode}
\usepackage{xcolor}
\definecolor{PaperSlate}{HTML}{586770}
\definecolor{PaperLinkBlue}{HTML}{2A7FB8}
\definecolor{PaperBlue}{HTML}{1F4E79}
\usepackage{titlesec}
\titleformat{\section}{\normalfont\Large\bfseries}{\thesection}{1em}{}
\titleformat{\subsection}{\normalfont\large\bfseries\color{PaperBlue}}{\thesubsection}{1em}{}
\titleformat{\subsubsection}{\normalfont\normalsize\bfseries\color{PaperBlue}}{\thesubsubsection}{1em}{}
\titleformat{\paragraph}[runin]{\normalfont\normalsize\bfseries\color{PaperBlue}}{\theparagraph}{1em}{}
\titlespacing*{\paragraph}{0pt}{1.7mm}{0.6em}
\usepackage[colorlinks=true,linkcolor=PaperLinkBlue,citecolor=PaperLinkBlue,urlcolor=PaperLinkBlue]{hyperref}
\newtheorem{theorem}{Theorem}[section]
\newtheorem{lemma}[theorem]{Lemma}
\newtheorem{proposition}[theorem]{Proposition}
\newtheorem{corollary}[theorem]{Corollary}
\newtheorem{assumption}[theorem]{Computational model}
\theoremstyle{definition}
\theoremstyle{remark}
\DeclareMathOperator{\Tr}{Tr}
\DeclareMathOperator{\ran}{ran}
\DeclareMathOperator{\Fid}{Fid}
\DeclareMathOperator{\diag}{diag}
\DeclareMathOperator{\rank}{rank}
\newcommand{\dd}{\mathrm d}
\newcommand{\R}{\mathbb R}
\newcommand{\C}{\mathbb C}

\newcommand{\norm}[1]{\lVert#1\rVert}
\newcommand{\ip}[2]{\langle#1,#2\rangle}
\numberwithin{equation}{section}
\allowdisplaybreaks[2]

\title{A Walk From Free Probability to Matrix\\
Discrepancy III: Higher Rank Kadison--Singer\\
and Spectrally Thin Trees}
\author{Tarun Kathuria\\[0.6ex]
\includegraphics[width=1.4in]{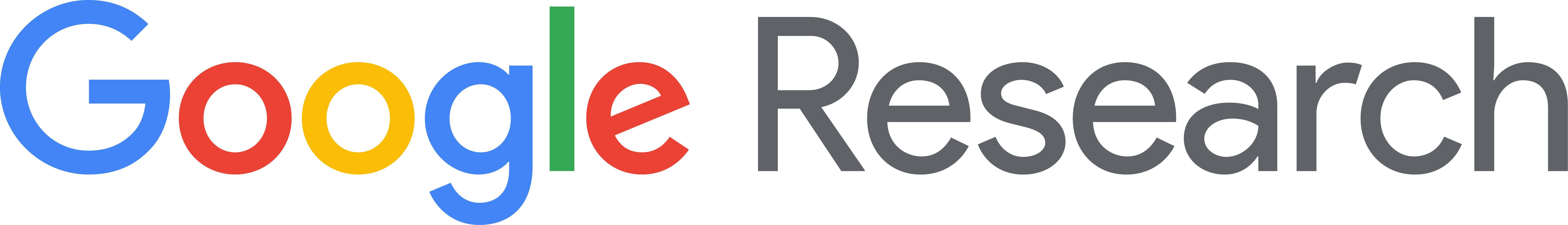}\\[0.6ex]
{\normalsize\href{mailto:tarunkathuria@google.com}{\texttt{tarunkathuria@google.com}}}}
\date{\today}
\hypersetup{pdftitle={A Walk From Free Probability to Matrix Discrepancy III: Higher Rank Kadison-Singer and Spectrally Thin Trees},pdfauthor={Tarun Kathuria}}
\begin{document}
\maketitle
\begin{abstract}
Let $A_1,\ldots,A_N$ be positive semidefinite matrices of rank at most
$r$, with $\sum_i A_i=I$ and $\norm{A_i}\le\varepsilon$. We prove
that one sign can be assigned to each original matrix with discrepancy
$O(\sqrt{\varepsilon\log(2r)})$, independently of their dimension and
number. We give a separate existence proof and a deterministic
algorithm with polynomial work in a real-arithmetic model with
semidefinite-value and exact spectral primitives. Separate Lean
formalizations verify the existence proof and the algorithm in this
arithmetic model.

The proof extends the variational approach to discrepancy developed
in the companion papers, motivated by operator-valued free
interpolation, Lehner's formula, and spectral Tsallis regularization.
A concave matrix power interpolates between a trace source and a
sandwich source. Its concavity controls the response of the
optimizing density. The walk maintains independent source reserves
and a second matrix recording reserve expenditure. This matrix
certifies contraction of the unfinished input mass between epochs;
within an epoch, preparation and a negative-curvature step control
the spectral potential while the coefficients advance toward signs.

As applications, one spanning tree can be chosen simultaneously
$O(\varepsilon\log(2s))$-spectrally thin for $s$ positive weightings
of a common graph whose edge leverages are at most $\varepsilon$.
The diagonal specialization gives discrepancy
$O(\sqrt{L\log(2k)})$ for matrices with row sums at most $L$ and at
most $k$ nonzero entries per column, through a walk of fractional
colorings.
\end{abstract}

\section{Introduction}
The rank-one discrepancy theorem underlying the solution of the
Kadison--Singer problem assigns signs to small positive matrices whose
sum is identity \cite{weaver2004,mss2015}. For higher-rank inputs, the
sign must be shared by every direction within an original matrix.
Splitting the matrices into rank-one summands and signing those
summands independently loses this constraint. We prove a bound whose
dependence on the rank is the square root of a logarithm.

\begin{theorem}[Higher-rank signing]\label{hr:thm:main}
Let $A_1,\ldots,A_N\in\C^{D\times D}$ be Hermitian positive
semidefinite matrices with
\[
 \sum_i A_i=I_D,\qquad \norm{A_i}\le\varepsilon,
 \qquad \rank A_i\le r,
\]
where $N,D,r\ge1$ and $\varepsilon>0$. There are signs
$\chi_i\in\{-1,1\}$ such that
\[
 \left\|\sum_i\chi_i A_i\right\|
 \le\min\{1,10^4\sqrt{\varepsilon\log(2r)}\}.
\]
In Model~\ref{hr:model:computation}, these signs can be found
deterministically with polynomial work in $N,D$. The existence
assertion holds without any computational assumption.
\end{theorem}

All logarithms without a subscript are natural. The constant in the
theorem covers both real and complex inputs; the proof for real
inputs gives $6000$ in the algorithmic assertion. The separate
existence proof gives a smaller constant. We have not optimized these
constants. Separate Lean formalizations follow the reserve epochs and
variational analysis. The existence theorem assumes only the matrix
hypotheses above. The algorithmic theorem verifies a fixed algorithm's
output and a uniform polynomial bound on its work in the following
arithmetic model, with the SDP solver supplied as an explicit parameter.

\begin{assumption}[Arithmetic and value primitives]\label{hr:model:computation}
Real arithmetic, comparisons, and nonnegative scalar square roots
have unit cost. Fix a deterministic solver and constants $A_*,b_*$
such that every feasible, finite-valued SDP constructed here, with
scalar data size $L$, has an additive-$\nu$ value report obtainable
in at most $A_*(L+1+\nu^{-1})^{b_*}$ operations. This guarantee
includes singular feasible programs and materialization of their explicit
affine block descriptions as real-coordinate SDPs. The input matrix
factors, centers, and reserve coefficients are computed separately.
An exact spectral primitive
returns an orthogonal eigendecomposition of a real symmetric matrix,
with a fixed basis and order at ties. An invocation has unit cost;
forming its input, reading its output, and subsequent arithmetic are
counted.
\end{assumption}

Section~\ref{hr:sec:sdp} derives the value programs from their matrix
variational formula. The polynomial bound is in this real-arithmetic
model, as in the companion rank-one paper \cite{kathuriaPartII}.
Binary encodings and implementations of the two primitives are outside
the statement.

\subsection{The source and the reserve}

The proof walks a coefficient vector $x$ from the origin of the cube
to a signing. A regularized spectral potential controls the matrix
increment made by the walk. The potential contains a positive source
that can be reduced as the coefficients move: removing source pays
for the curvature created by changing the discrepancy matrix. The
main issue is to use this reserve efficiently while still ensuring
that every coordinate eventually reaches a face.

There are two natural sources for a positive matrix $A_i$ and a
trace-one density $T$:
\[
 \Tr(A_iT)A_i\qquad\hbox{and}\qquad A_iTA_i.
\]
They agree at rank one. The first has a simple density response, but
its trace can cost $r\varepsilon$. The second has a smaller trace
budget, but its response allows density mass to move inside the
range of $A_i$. We interpolate between them using
\begin{equation}\label{hr:eq:introsource}
 M_i=A_i^{1/2}TA_i^{1/2},\quad p_i=\Tr M_i,\qquad
 A_i^{1/2}\bigl[p_i^\beta M_i^{1-\beta}\bigr]A_i^{1/2},
 \qquad \beta\asymp\frac1{\log(2r)}.
\end{equation}
The trace of the middle factor is at most $r^\beta p_i$.
Its internal matrix response is controlled by its own concavity,
at a cost proportional to $1/\beta$. This is the local analytic
estimate at the center of the proof.

To turn this estimate into the desired discrepancy bound, the source
reserve must have initial height only $O(1/\beta)$. We therefore
maintain a reserve $c_i$ independently of $x_i$, together with a
counter $s_i$ for the reserve already spent. During an epoch,
$c_i+a s_i=aR$ for every interior coordinate, where $R=4$ and
$a=O(1/\beta)$. There are two ways to spend reserve. A preparation
step decreases $c_i$ while increasing $s_i$, leaving $x$ fixed.
A walk step has the form
\[
 x\longmapsto x\pm th,\qquad
 s_i\longmapsto s_i+t^2h_i^2,\qquad
 c_i\longmapsto c_i-a t^2h_i^2.
\]
Both operations preserve the reserve identity. The counter is needed
because a coordinate might use up its reserve before reaching a face.

The potential consequently controls two matrices. If $x^0$ is the
coefficient vector at the start of the epoch, they are
\begin{equation}\label{hr:eq:introcenters}
 H=\sum_i(x_i-x_i^0)A_i,\qquad
 K=\sum_i\bigl((x_i^0)^2-x_i^2+s_i\bigr)A_i.
\end{equation}
The first is the discrepancy increment. The second certifies that
spending reserve has made progress: once all reserves vanish, every
unfinished coordinate has $s_i=R$, so its original matrix is charged
to a positive budget bounded by the initial unfinished mass and $K$.
This makes the norm of the unfinished mass contract between epochs.
The fractional coordinates persist through a restart; only the reserve
counters and the reference point $x^0$ are reset.

There is also a local benefit to the second matrix. In the walk
update, the quadratic change of $-x_i^2$ is canceled exactly by the
increase of $s_i$. Thus both centers in \eqref{hr:eq:introcenters}
vary affinely along a proposed step. The source decreases to second
order and pays for the response of the optimizing density. A
projection argument supplies a direction perpendicular to $x$;
both signs increase $\norm{x}_2^2$, and one sign decreases the
spectral potential. Preparation ensures the hypothesis of that
projection argument. This gives the local walk, while contraction
of the unfinished matrix mass gives global completion.

\subsection{Variational motivation and related bounds}

The companion papers develop discrepancy potentials through
operator-valued free interpolation, Lehner's variational formula, and
spectral density regularization
\cite{kathuriaPartI,kathuriaPartII,bbvh2023,lehner1999}. The
square-root density term is the spectral Tsallis--$1/2$ regularizer
used in spectral sparsification framed as matrix optimization
\cite{allenZhuLiaoOrecchia2015} and in discrepancy minimization
\cite{pesentivladu2026}. The higher-rank source
\eqref{hr:eq:introsource} is nonlinear in the optimizing density;
we analyze its variational formula directly. In particular, we do
not identify it with a fixed linear operator-valued covariance map.

The variational optimizer satisfies coupled equations for a density
and a positive transport matrix. Differentiating these equations
determines their response to a proposed coefficient movement.
This shares the stability viewpoint of matrix Dyson equation
analysis \cite{aek2019,aeks2020}: a matrix equation must be controlled
in the geometry selected by its positive solution. Here the regularized
optimizer selects that solution, and a compatible subspace of
coefficient forces suffices. Section~\ref{hr:sec:curvature} derives
the linearized equations and constructs that subspace.

Trace-based higher-rank extensions of the interlacing-polynomial
method give dimension-free estimates controlled by the largest input
trace \cite{cohen2016higherRank,branden2018hyperbolic,xuXuZhu2021}.
Under the hypotheses of Theorem~\ref{hr:thm:main}, that gives the
scale $O(\sqrt{r\varepsilon})$ in the nontrivial regime. Our
bound replaces this rank dependence by $\sqrt{\log(2r)}$.

The shared sign within each matrix also gives a simultaneous graph
application. For $s$ positive weightings of a common connected graph,
put the $s$ normalized edge matrices into one direct-sum matrix per
edge. Its rank is at most $s$. Repeated signing then gives one
spanning tree that is $O(\varepsilon\log(2s))$-spectrally thin for
every weighting, provided each edge leverage is at most
$\varepsilon$ in each weighting. The single-weighting existence
bound is known \cite{harveyOlver2014}; the simultaneous reduction
and its algorithmic form are proved in Section~\ref{hr:sec:thin-trees}.

For diagonal matrices, the theorem gives discrepancy
$O(\sqrt{L\log(2k)})$ when rows have sum at most $L$ and each
column has at most $k$ nonzero entries. When both row and column
sparsity are at most $t$, this is $O(\sqrt{t\log(2t)})$,
independent of the numbers of rows and columns. It reaches the
usual local-lemma scale in this regime
\cite{harvey2014rowcolumn} through a walk of fractional colorings.
Section~\ref{hr:sec:beck-fiala} gives the scalar potential and
discusses the relation to iterative regularization and other
discrepancy walks \cite{pesentivladu2026,bansalDadushGarg2019,
bansalJiang2025decoupling}.

\paragraph{Organization.}
Sections~\ref{hr:sec:setup}--\ref{hr:sec:algorithm} define the
potential, state, and algorithm. The following three sections derive
the variational response and the local movement estimate. We first
use them to give a compact existence proof, then prove the finite
algorithm and its polynomial work bound. The two applications follow;
quantitative derivative estimates are collected in the appendix.

\section{State and potential}\label{hr:sec:setup}

All traces are unnormalized. Matrix norms without a subscript are
operator norms, $\norm{\cdot}_{\mathrm F}$ is the Frobenius norm,
and the real inner product on Hermitian matrices is
$\ip XY=\operatorname{Re}\Tr(XY)$. Write $\ran A$ for the range
of $A$, and $A^*$ for its conjugate transpose. A label always denotes
one original input matrix, including its entire range.

The analytic argument works for a subisotropic family:
\begin{equation}\label{hr:eq:input}
 A_i\succeq0,\qquad B:=\sum_i A_i\preceq I_D,\qquad
 \norm{A_i}\le\varepsilon,\qquad \rank A_i\le r.
\end{equation}
Zero matrices may be removed and assigned arbitrary signs at the end.
Choose $0<\beta\le1/2$ with $r^\beta\le2$, and put
$\alpha=1-\beta$. The finite algorithm uses
\begin{equation}\label{hr:eq:qchoice}
 q=\text{the least power of two at least }
       \max\{2,\log_2(2r)\},\qquad \beta=q^{-1}.
\end{equation}
Then $q\le2\log_2(2r)$. Dyadic $q$ permits an exact SDP
representation by iterated square roots.

\subsection{One epoch}

An epoch begins with the unfinished original labels and their current
coefficients $x_i^0\in(-1,1)$. Its unfinished mass is $B=\sum_i A_i$
and its size is $b=\norm B$. Fix $R=4$ and $a>0$. The state
consists of coefficients $x_i$, spent reserves $s_i$, and remaining
reserves $c_i$, subject to
\begin{equation}\label{hr:eq:state}
\begin{gathered}
 -1\le x_i\le1,\quad 0\le s_i\le R,\quad 0\le c_i\le aR,\\
 c_i+a s_i\le aR,\qquad
 (1-x_i^2)(c_i+a s_i-aR)=0.
\end{gathered}
\end{equation}
The initial state is $(x,s,c)=(x^0,0,aR\mathbf1)$.
An interior coordinate has $c_i+a s_i=aR$. At a face, unused
reserve can be deleted without changing $s_i$. A zero reserve at
an interior coordinate means that this coordinate waits for the
next epoch; it has not yet been assigned its final sign.

Define the two centers and the nonnegative budget matrix by
\begin{equation}\label{hr:eq:centers}
\begin{aligned}
 H&=\sum_i(x_i-x_i^0)A_i,&
 K&=\sum_i\bigl((x_i^0)^2-x_i^2+s_i\bigr)A_i,\\
 \mathcal B&=\sum_i(1-x_i^2+s_i)A_i
       =\sum_i(1-(x_i^0)^2)A_i+K.
\end{aligned}
\end{equation}
Both centers vanish initially. Throughout the feasible set,
$\norm H\le2b$ and $\norm K\le(R+1)b$. If all $c_i$ vanish,
the unfinished mass $B_{\rm next}=\sum_{|x_i|<1}A_i$ satisfies
\begin{equation}\label{hr:eq:budgetcompletion}
 R B_{\rm next}\preceq\mathcal B\preceq B+K.
\end{equation}
Indeed, each summand of $\mathcal B$ is positive semidefinite,
and an unfinished coordinate has $s_i=R$. This is the reason for
including $K$ in the spectral potential.

\subsection{Four blocks and the power source}

The four-block center controls both signs of both matrices:
\[
 \mathcal H=\diag(H,-H,K,-K),\qquad B_i=I_4\otimes A_i.
\]
A density is a matrix $S\succeq0$ on $\C^4\otimes\C^D$ with
$\Tr S=1$. Denote its blocks by $S_{jk}$, $1\le j,k\le4$,
and its marginal by $T=\sum_{j=1}^4 S_{jj}$. On the fixed carrier
$\mathcal K_i=\ran A_i$, put
\begin{equation}\label{hr:eq:carrier}
 M_i=A_i^{1/2}TA_i^{1/2},\qquad p_i=\Tr M_i,\qquad
 F_\beta(M)=(\Tr M)^\beta M^\alpha,
 \quad F_\beta(0)=0.
\end{equation}
Carrier powers are extended by zero to the ambient space. Define
\begin{equation}\label{hr:eq:source}
 \Omega_i(S)=I_4\otimes
       A_i^{1/2}F_\beta(M_i)A_i^{1/2},\qquad
 \Omega_c(S)=\sum_i c_i\Omega_i(S).
\end{equation}
The same marginal enters every copy of a source. At rank one,
$F_\beta(M_i)=M_i$ and $\Omega_i(S)=p_i B_i$.

For positive semidefinite $S,W$, their fidelity is
$\Fid(S,W)=\Tr\sqrt{S^{1/2}WS^{1/2}}$. With $\theta>0$,
the value function is
\begin{equation}\label{hr:eq:potential}
 \mathcal E_\theta(H,K,c)=
 \max_{\substack{S\succeq0\\\Tr S=1}}
 \left\{\Tr(\mathcal H S)+2\Fid(S,\Omega_c(S))
                         +2\theta\Tr\sqrt S\right\}.
\end{equation}
We also write $\mathcal E_\theta(x,s,c)$ after substituting
\eqref{hr:eq:centers}. The root term keeps the maximizing density
positive definite. The two basic potential bounds are
\begin{equation}\label{hr:eq:initialbudget}
 \max\{\norm H,\norm K\}\le\mathcal E_\theta(H,K,c),\qquad
 \mathcal E_\theta(0,0,aR\mathbf1)
 \le4\sqrt{aR\varepsilon r^\beta b}+2\theta\sqrt{4D}.
\end{equation}
They are proved in Lemma~\ref{hr:lem:budget}. With
$a=O(1/\beta)$, the initial spectral cost is
$O(\sqrt{\varepsilon b/\beta})$.

\subsection{The two local operations}

Preparation of label $i$ replaces $(c_i,s_i)$ by
$(c_i-t,s_i+t/a)$ with $x$ fixed. It reduces source while changing
$K$ by $tA_i/a$. The value test in the algorithm measures this
tradeoff. If every small preparation is unsuccessful, the derivative
test proved below bounds the cost of moving each active coordinate.

For a vector $h$ supported on the positive-reserve coordinates, the
walk uses
\begin{equation}\label{hr:eq:movement}
 x_i^\pm=x_i\pm th_i,\qquad
 s_i^\pm=s_i+t^2h_i^2,\qquad
 c_i^\pm=c_i-at^2h_i^2.
\end{equation}
Writing $A(h)=\sum_i h_i A_i$ and $(xh)_i=x_i h_i$, direct
expansion gives the exact identities
\begin{equation}\label{hr:eq:center-cancellation}
 H^\pm=H\pm tA(h),\qquad K^\pm=K\mp2tA(xh).
\end{equation}
Thus the positive source reduction is quadratic, while the center
change is affine. The local response estimate finds a direction for
which this source reduction outweighs the second-order spectral cost.

\section{The finite walk}\label{hr:sec:algorithm}

The algorithm works on the unfinished labels in epochs. Within an
epoch it first removes coordinates close to faces and reserves close
to zero. It then tests whether a small preparation decreases the
potential. If every preparation fails, it computes a negative
curvature direction perpendicular to the current coefficient vector
and chooses the better of its two signs. An epoch ends when every
reserve has vanished. The budget matrix then shows that the total
unfinished matrix mass has decreased, so another epoch can be
started at smaller spectral cost.

We give the real symmetric algorithm first. Complex inputs are
handled by realifying each entire input matrix, as described at the
end of Section~\ref{hr:sec:runtime}. Fix all tie rules once, including
the order of preparation tests. The case $D=0$ returns all positive
signs. Otherwise assume $\norm{\sum_i A_i}=1$ and
$\sum_i A_i\preceq I_D$. Compute the actual maximum norm
$\varepsilon=\max_i\norm{A_i}$ and maximum rank $r$ by the
spectral primitive. Then $1/N\le\varepsilon\le1$ and $1\le r\le D$.
These choices only improve a supplied upper bound.

\subsection{Parameters and the value routine}

Choose $q,\beta$ by \eqref{hr:eq:qchoice}. This requires no
logarithm evaluation: start with $(q,w)=(2,4)$ and, while $w<2r$,
replace $(q,w)$ by $(2q,w^2)$. The invariant $w=2^q$ shows that
the returned $q$ is the required least dyadic scale. Set
\begin{equation}\label{hr:eq:algorithmparameters}
\begin{gathered}
 a=1024q,\quad R=4,\quad T_*=1+N+D+q,\qquad
 \theta=\eta=\rho=\zeta=T_*^{-10},\\
 \alpha_0=2aR\varepsilon,\qquad \delta=5\sqrt{\alpha_0},\qquad
 \overline B=64aT_*,\qquad s_0=(\theta/\overline B)^2,\\
 p_0=s_0\eta,\qquad \tau_0=4s_0\eta^2/\overline B,
 \qquad \gamma=a\tau_0/4.
\end{gathered}
\end{equation}
Here $s_0$ is a lower bound for density eigenvalues, $p_0$ and
$\tau_0$ are lower bounds for the source probes defined in
Section~\ref{hr:sec:curvature}, and $\gamma$ is the curvature
scale. Let $M$ be the explicit derivative bound in
\eqref{hr:eq:derivative-recipe}. It is computed from the parameters
above by arithmetic and square roots. Set
\begin{equation}\label{hr:eq:stepsizes}
\begin{gathered}
 r_* =\min\{\rho/4,\sqrt{\zeta/(4a)}\},\qquad
 \lambda=\min\{\zeta,p_0/(8aM)\},\\
 t=\min\{r_*/4,\gamma/(64NM)\},\qquad
 h=\min\{r_*/4,\gamma/(16M)\},\\
 \nu=\min\left\{1,\frac{\lambda p_0}{64a},
                       \frac{\gamma t^2}{256N},
                       \frac{\gamma h^2}{64}\right\}.
\end{gathered}
\end{equation}
The preparation length is $\lambda$, the finite-difference spacing
is $t$, and the walk length is $h$. Every value report has absolute
error at most $\nu$. The appendix proves the required derivative
bounds; Section~\ref{hr:sec:runtime} shows that all inverse
accuracies and call counts are polynomial in $N,D$.

If $\delta\ge1$, return all positive signs. Otherwise set aside
every whole matrix of norm at most $\eta$, assigning it sign $+1$.
Their combined discrepancy is at most $N\eta$. The remaining
matrices form a subisotropic family. The original $N,D,q$ continue
to determine all parameters.

The routine $\operatorname{VALUE}(H,K,c)$ returns an additive-$\nu$
report for \eqref{hr:eq:potential}, with the current epoch family
and the fixed value of $\theta$. Section~\ref{hr:sec:sdp} gives
its complete SDP. For each current state define
\begin{equation}\label{hr:eq:compositevalue}
 \varphi(u)=\mathcal E_\theta
       \bigl(H+A(u),\ K-2A(xu),\ c-au^{\odot2}\bigr).
\end{equation}
Here $u$ is supported on the active set $I=\{i:c_i>0\}$ and
$u^{\odot2}$ denotes coordinatewise squaring. Equation
\eqref{hr:eq:center-cancellation} identifies this expression with
the value of the actual update of $x,s,c$.

\subsection{The Hessian and the two signs}

Let $V(u)$ be the report obtained by applying $\operatorname{VALUE}$
to the three arguments in \eqref{hr:eq:compositevalue}. For active
labels $i,j$, form the real symmetric matrix $\widetilde L$ by
\begin{equation}\label{hr:eq:hessianstencil}
\begin{aligned}
 \widetilde L_{ii}
 &=\frac{V(te_i)-2V(0)+V(-te_i)}{t^2},\\
 \widetilde L_{ij}
 &=\frac{V(te_i+te_j)-V(te_i-te_j)
                 -V(-te_i+te_j)+V(-te_i-te_j)}{4t^2}\quad(i\ne j).
\end{aligned}
\end{equation}
Write $z=(x_i)_{i\in I}$ and let
\[
 P_x=\begin{cases}I,&z=0,\\ I-zz^\top/\norm z_2^2,&z\ne0.
 \end{cases}
\]
This projects active directions onto $x^\perp$. On the curvature
branch, the smallest eigenvalue of $P_x\widetilde L P_x$ is
strictly negative. A unit eigenvector $g$ for it is therefore in
$\ran P_x$; extend $g$ by zero outside $I$. Both signs satisfy
\begin{equation}\label{hr:eq:radialprogress}
 \norm{x\pm hg}_2^2=\norm x_2^2+h^2.
\end{equation}
The value comparison chooses which sign decreases the spectral
potential. Progress toward the faces holds for either choice.
Equivalently, one may form an orthonormal basis $U$ of $x^\perp$,
diagonalize $U^\top\widetilde L U$, and lift its minimum eigenvector
through $U$. This is the implementation used in the certificate;
a Householder reflection constructs $U$ using polynomially many
scalar operations.

Figure~\ref{hr:fig:tangent} illustrates the coefficient movement and
the reserve update shared by its two candidates.
\begin{figure}[ht]
\centering
\begin{minipage}{0.40\textwidth}
\centering
\begin{tikzpicture}[scale=1.9,>=Latex,
                    every node/.style={font=\small}]
 \draw[PaperSlate!65] (-1,-1) rectangle (1,1);
 \draw[PaperSlate!35,dashed] (0,0) circle (0.6);
 \draw[->,PaperSlate] (0,0) -- (0.6,0);
 \draw[->,thick,PaperBlue] (0.6,0) -- (0.6,0.6);
 \draw[->,thick,PaperBlue] (0.6,0) -- (0.6,-0.6);
 \draw[PaperSlate] (0.50,0) -- (0.50,0.10) -- (0.60,0.10);
 \fill[PaperSlate] (0,0) circle (0.025) node[below left] {$0$};
 \fill[PaperBlue] (0.6,0) circle (0.025) node[right] {$x$};
 \fill[PaperBlue] (0.6,0.6) circle (0.025)
       node[above left] {$x+hg$};
 \fill[PaperBlue] (0.6,-0.6) circle (0.025)
       node[below left] {$x-hg$};
\end{tikzpicture}
\end{minipage}\hfill
\begin{minipage}{0.55\textwidth}
\small
The active unit direction satisfies $g\perp x$, so
\[
 \norm{x\pm hg}_2^2=\norm x_2^2+h^2.
\]
Both candidates use the same reserve update:
\[
 s'=s+h^2g^{\odot2},\qquad c'=c-ah^2g^{\odot2}.
\]
Comparing their value reports selects a decrease of the joint
discrepancy--budget potential; radial progress holds for either sign.
\end{minipage}
\caption{The higher-rank curvature step in two coefficient dimensions.
The dashed circle is a level set of the squared coefficient norm.
The direction is zero outside the active set, and the step-size bound
keeps both candidates feasible. Preparation instead spends reserve
with $x$ fixed. Exhausting a coordinate's reserve suspends its movement
for the current epoch; near-face rounding fixes its sign permanently.}
\label{hr:fig:tangent}
\end{figure}

\begin{algorithm}[H]
\caption{Higher-rank signing with independent reserves}\label{hr:alg:walk}
\begin{algorithmic}[1]
\Require Real PSD inputs with $\sum_i A_i\preceq I_D$ and
         $\norm{\sum_i A_i}=1$.
\State Compute \eqref{hr:eq:algorithmparameters}--\eqref{hr:eq:stepsizes}.
\If{$\delta\ge1$}\State \Return all positive signs.\EndIf
\State Set aside matrices of norm at most $\eta$ with sign $+1$;
       set every retained coefficient to zero.
\While{$\norm{\sum_{i:\,|x_i|<1}A_i}>\delta^2$}
 \State Begin an epoch on the unfinished labels:
        $x^0\gets x$, $s\gets0$, $c\gets aR\mathbf1$.
 \While{some $c_i>0$}
  \State Round each unfinished $i$ with $1-|x_i|\le\rho$
         outward to a sign; set $c_i\gets0$, leaving $s_i$ fixed.
  \State For each remaining interior $i$ with $0<c_i<2\zeta$,
         set $s_i\gets s_i+c_i/a$, then $c_i\gets0$.
  \If{all $c_i=0$}\State \textbf{break}.\EndIf
  \State Form $H,K$ by \eqref{hr:eq:centers};
         $v_0\gets\operatorname{VALUE}(H,K,c)$.
  \For{active $i$ in the fixed order}
   \State $v_i\gets\operatorname{VALUE}(H,K+(\lambda/a)A_i,c-\lambda e_i)$.
   \If{$v_i-v_0\le-\lambda p_0/(4a)$}
    \State $c_i\gets c_i-\lambda$, $s_i\gets s_i+\lambda/a$;
           restart the inner while-loop.
   \EndIf
  \EndFor
  \State Form $\widetilde L$ by \eqref{hr:eq:hessianstencil}.
         Compute a minimum unit eigenvector $g$ of $P_x\widetilde L P_x$.
  \State Choose $\omega\in\{-1,1\}$ minimizing $V(\omega hg)$.
  \State $x\gets x+\omega hg$, $s\gets s+h^2g^{\odot2}$,
         $c\gets c-ah^2g^{\odot2}$.
 \EndWhile
 \State Keep all completed signs and all unfinished fractional coordinates.
\EndWhile
\State Round the remaining coefficients to signs; return these and the set-aside signs.
\end{algorithmic}
\end{algorithm}

All comparisons use fixed tie rules. In particular, a successful
preparation invalidates earlier rejected tests, so the scan restarts
at the new state. A coordinate with exhausted reserve stays fixed
for the rest of its epoch and is active again in the next epoch if
it is still fractional.

\begin{figure}[ht]
\centering
\begin{tikzpicture}[>=Latex,node distance=9mm,
 every node/.style={font=\small,align=center},
 box/.style={draw=PaperSlate,rounded corners,inner sep=5pt}]
 \node[box] (prep) {Try reserve preparation\\$x$ fixed, $c_i\downarrow$, $s_i\uparrow$};
 \node[box,right=13mm of prep] (walk) {If every test fails: $g\perp x$\\choose the better of $x\pm hg$};
 \node[box,below=of prep] (end) {All reserves exhausted\\$4B_{\rm next}\preceq B+K$};
 \node[box,right=13mm of end] (epoch) {Restart on unfinished labels\\retain $x$, reset reserves};
 \draw[->] (prep)--(walk);
 \draw[->] (walk.south)--++(0,-3mm)-|(prep.south);
 \draw[->] (prep)--(end);
 \draw[->] (end)--(epoch);
\end{tikzpicture}
\caption{Local descent controls an epoch; the budget matrix pays for
the next one. The actual algorithm also makes the small cleanup
operations displayed in Algorithm~\ref{hr:alg:walk}.}
\end{figure}
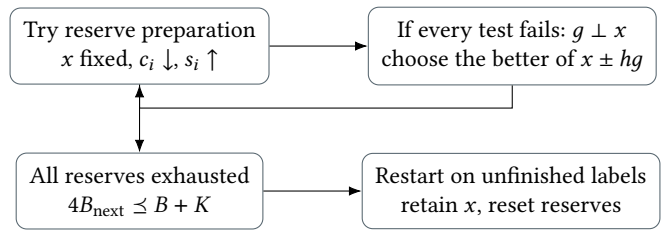

\section{The variational potential}\label{hr:sec:foundations}

\subsection{Concavity and supported transport}

We first establish the analytic properties used throughout the proof.
The same resolvent formula will later supply the crucial source
curvature estimate.

\begin{lemma}[The carrier power]\label{hr:lem:carrierpower}
The function $F_\beta$ is continuous on the positive semidefinite
cone, homogeneous of degree one, and operator concave. It is smooth
on the positive definite cone of its fixed carrier. For
$f(M)=M^\alpha$, $M\succ0$, and Hermitian $U,V$, let
$R_t=(M+tI)^{-1}$ and $a_\alpha=\sin(\pi\alpha)/\pi$. Then
\begin{align}
 Df(M)[U]
 &=a_\alpha\int_0^\infty t^\alpha R_tUR_t\,\dd t,
                                                        \label{hr:eq:powerderivative}\\
 -D^2f(M)[U,V]
 &=a_\alpha\int_0^\infty t^\alpha
       R_t(UR_tV+VR_tU)R_t\,\dd t.           \label{hr:eq:powersecond}
\end{align}
In particular, $Df(M)$ preserves positive semidefiniteness.
\end{lemma}
\begin{proof}
The scalar identity
\[
 \lambda^\alpha
 =a_\alpha\int_0^\infty
       t^{\alpha-1}\frac{\lambda}{\lambda+t}\,\dd t
 \qquad(\lambda>0)
\]
follows by substituting $t=\lambda s$ in the Euler beta integral.
Applying the spectral theorem gives the corresponding matrix
identity. Differentiating $M(M+tI)^{-1}=I-t(M+tI)^{-1}$ yields
\eqref{hr:eq:powerderivative} and
\eqref{hr:eq:powersecond}. Differentiation under the integral is
valid locally on $M\succ0$: on a neighborhood with a fixed positive
lower eigenvalue bound, the first-derivative integrand is
$O(t^\alpha)$ at zero and $O(t^{\alpha-2})$ at infinity, and the
second derivative decays still faster at infinity.
Higher derivatives satisfy analogous integrable local bounds.

The inverse map is operator convex on positive definite matrices.
For example, its second derivative in a Hermitian direction $U$ is
$2RURUR\succeq0$, where $R$ is the inverse; this is a congruence of
the square of $R^{1/2}UR^{1/2}$. Hence
$M(M+tI)^{-1}=I-t(M+tI)^{-1}$ is operator concave, and integration
proves operator concavity of $M^\alpha$. Continuity extends it to
singular matrices.

To pass from the power to $F_\beta$, write
$F_\beta(M)=p(M/p)^\alpha$ when $p=\Tr M>0$. Given nonzero
$M_1,M_2$ and $0<\lambda<1$, put
$p=\lambda p_1+(1-\lambda)p_2$. The normalized convex combination
has weights $\lambda p_1/p$ and $(1-\lambda)p_2/p$. Applying
operator concavity of the power to that combination and multiplying
by $p$ proves
\[
 F_\beta(\lambda M_1+(1-\lambda)M_2)
 \succeq\lambda F_\beta(M_1)+(1-\lambda)F_\beta(M_2).
\]
Zero inputs follow by homogeneity and continuity. At zero,
$\norm{F_\beta(M)}\le\Tr M$; away from zero, continuity follows
from continuity of the matrix power and scalar trace factor.
Homogeneity is immediate from $\alpha+\beta=1$.
\end{proof}

Each $\Omega_i$ is therefore operator concave, homogeneous of
degree one, and continuous. To combine this fact with fidelity,
we use the following finite-dimensional transport formula.

\begin{lemma}[Fidelity and transport]\label{hr:lem:transport}
Let $W\succeq0$ be supported on a subspace $\mathcal K$, and let
$S_0$ be the compression of $S\succeq0$ to $\mathcal K$. Then
\begin{equation}\label{hr:eq:transportvariational}
 2\Fid(S,W)=
 \inf_{\substack{Z\succ0\\\text{on }\mathcal K}}
       \{\Tr(S_0Z^{-1})+\Tr(WZ)\}.
\end{equation}
If $S_0,W\succ0$ on $\mathcal K$, the unique minimizer is
\begin{equation}\label{hr:eq:transport}
 Z=W^{-1/2}(W^{1/2}S_0W^{1/2})^{1/2}W^{-1/2},
 \qquad ZWZ=S_0.
\end{equation}
Fidelity is jointly concave and is increasing in either argument
in positive semidefinite order.
\end{lemma}
\begin{proof}
The nonzero eigenvalues in the definition of fidelity can equivalently
be computed from $W^{1/2}SW^{1/2}$, so the expression depends on $S$
only through $S_0$. Its value is
$\norm{S_0^{1/2}W^{1/2}}_*$, where $\norm{\cdot}_*$ is the sum
of singular values. The polar variational formula and
$2\operatorname{Re}\Tr(A^*B)\le\norm A_{\mathrm F}^2+
\norm B_{\mathrm F}^2$, applied with
$A=Z^{-1/2}S_0^{1/2}$ and $B=Z^{1/2}W^{1/2}U$ for a unitary
$U$, give the upper bound by every transport cost in
\eqref{hr:eq:transportvariational}.
For positive definite arguments, the matrix in
\eqref{hr:eq:transport} has both cost terms equal to
$\Tr\sqrt{W^{1/2}S_0W^{1/2}}$. It attains equality, and the
equation $ZWZ=S_0$ determines it uniquely.

For singular arguments, regularize both compressed arguments by
adding $\delta I_{\mathcal K}$. A transport optimizer for the
regularized pair has cost for the original pair no larger than
its regularized cost. Letting $\delta\downarrow0$ proves equality
in \eqref{hr:eq:transportvariational} by continuity.
On a common ambient space, that formula is an infimum of linear
functions of $(S,W)$ and hence is jointly concave. Every such
linear function is increasing in both positive semidefinite
arguments, which also proves monotonicity.
\end{proof}

\begin{proposition}[The optimizing density and its support]
\label{hr:prop:regularity}
For $\theta>0$, the value in \eqref{hr:eq:potential} is continuous
in the centers and nonnegative reserves. Its maximizing density is
unique and positive definite. On a region where the set
$I=\{i:c_i>0\}$ is fixed, the density and value are smooth in
the centers and positive reserves. The transport is smooth on the
fixed source support
\begin{equation}\label{hr:eq:livesupport}
 \mathcal K=\C^4\otimes\operatorname{span}_{i\in I}\ran A_i.
\end{equation}
Both the density and supported transport are block diagonal in the
four displayed blocks.
\end{proposition}
\begin{proof}
Operator concavity of the source, monotonicity of fidelity, and its
joint concavity make the density objective concave. Its root term
is strictly concave. To verify strictness, at $S\succ0$ diagonalize
$Q=\sqrt S$ with eigenvalues $\lambda_j>0$. Differentiating
$Q^2=S$ gives $Q\,DQ[X]+DQ[X]Q=X$, hence
\begin{equation}\label{hr:eq:rootcurvature}
 -D^2(2\theta\Tr\sqrt S)[X,X]
 =\theta\sum_{j,k}
 \frac{|X_{jk}|^2}{\lambda_j\lambda_k(\lambda_j+\lambda_k)}>0
 \quad(X\ne0).
\end{equation}
For a segment between singular densities, restrict to the range of
their sum, where its interior is positive definite. This proves
strict concavity on the entire density set. Compactness gives a
unique maximizer.

If a maximizing density were singular, mix it with $I_{4D}/(4D)$
with weight $t$. Concavity bounds the loss in the other terms by
$O(t)$, whereas a zero eigenvalue contributes a positive root gain
of order $\sqrt t$. The maximizer is therefore positive definite.
At such a density, each $M_i$ is positive definite on $\ran A_i$,
and $\Omega_c(S)$ is positive definite on \eqref{hr:eq:livesupport}.
The compressed fidelity and the full density objective are smooth
there. The Hessian on the trace-zero density tangent is strictly
negative by \eqref{hr:eq:rootcurvature}. Apply the smooth implicit
function theorem to the density stationarity equations: a smooth
equation $G(u,S)=0$ with invertible derivative in $S$ has a unique
smooth local solution $S=S(u)$ through the given point
\cite[Appendix~B]{guilleminhaine2019}. Uniqueness identifies it
with the optimizing density; \eqref{hr:eq:transport} then gives
smooth transport. Joint continuity on the compact density domain
also proves continuity of the optimized value when reserves vanish.

Finally, conjugation by independent scalar signs on the four blocks
preserves the center and marginal. Averaging over these conjugations
removes off-diagonal blocks and cannot decrease the concave objective.
Uniqueness makes the optimizer block diagonal. Its compressed source
and the explicit transport formula have the same symmetry.
\end{proof}

The density remains a full ambient matrix when $\mathcal K$ is a
proper subspace. Compression applies to the source and transport,
not to the density domain. Write $\pi$ for compression to
$\mathcal K$ and $\iota$ for extension by zero from that space;
these maps are adjoint under the trace pairing. With $S_0=\pi S$,
the transport and density stationarity equations are
\begin{equation}\label{hr:eq:stationarity}
\begin{split}
 Z\Omega_c(S)Z&=S_0,\\
 \mathcal H+\iota(Z^{-1})+D\Omega_c(S)^*(Z)
                  +\theta S^{-1/2}&=\ell I_{4D},\qquad \Tr S=1.
\end{split}
\end{equation}
The adjoint in this formula is the adjoint of the actual source
derivative, extended to the full density space. It need not equal
the derivative itself. The equation follows by differentiating the
transport representation at its optimizer and applying the Lagrange
multiplier rule on the trace-one affine space. Thus every term is
determined by the same optimizing pair.

\begin{lemma}[Spectral and initial bounds]\label{hr:lem:budget}
The bounds in \eqref{hr:eq:initialbudget} hold. The value is
nondecreasing in each reserve and is $1$-Lipschitz in the operator
norm of the four-block center.
\end{lemma}
\begin{proof}
A rank-one density on a top eigenvector of $\mathcal H$ gives
$\mathcal E_\theta\ge\lambda_{\max}(\mathcal H)
=\max\{\norm H,\norm K\}$. For $M\succeq0$ of trace $p$
and rank at most $r$, scalar concavity on its eigenvalues gives
$\Tr F_\beta(M)\le r^\beta p$. At the initial point,
\[
 \Tr\Omega_{aR}(S)
 \le4aR\varepsilon r^\beta\sum_i p_i
 =4aR\varepsilon r^\beta\Tr(TB)
 \le4aR\varepsilon r^\beta b.
\]
Apply $\Fid(S,W)\le\sqrt{\Tr S\,\Tr W}$ and
$\Tr\sqrt S\le\sqrt{4D}$ to prove the upper bound.
Increasing a reserve increases the source in PSD order, so fidelity
monotonicity proves the first assertion. For the second, every density
satisfies $|\Tr((\mathcal H'-\mathcal H)S)|
\le\norm{\mathcal H'-\mathcal H}$; maximize over the same domain.
\end{proof}

\section{How source concavity controls its derivative}\label{hr:sec:sourcemetric}

The source depends on the density, so changing the density changes
both arguments of fidelity. The derivative of an individual carrier
has a scalar part, corresponding to its change of trace, and a
remaining matrix part. This section bounds the latter by the
carrier's own concavity. The scalar estimate illustrates the mechanism. The matrix estimate
then controls the complete output in the quadratic norm that arises
when the transport is eliminated.

Throughout this section, $M\succ0$ acts on one fixed carrier $K$,
$p=\Tr M$, and $F_\beta(M)=p^\beta M^\alpha$. All variations
are Hermitian. The dimension of this carrier is arbitrary.

\subsection{Removing the scalar trace direction}

For a variation $U$, put
\[
 a=\frac{\Tr U}{p},\qquad U_0=U-aM.
\]
Then $\Tr U_0=0$. Homogeneity of degree one gives
\begin{equation}\label{hr:eq:centering}
 DF_\beta(M)[M]=F_\beta(M),\qquad
 D^2F_\beta(M)[M,V]=0,
\end{equation}
and hence
\[
 DF_\beta(M)[U]-aF_\beta(M)=DF_\beta(M)[U_0],
 \qquad
 D^2F_\beta(M)[U,U]=D^2F_\beta(M)[U_0,U_0].
\]
The second identity in \eqref{hr:eq:centering} follows by
differentiating the first in direction $V$. Along the centered
direction the trace factor is constant, so these two expressions
are $p^\beta D(M^\alpha)[U_0]$ and
$p^\beta D^2(M^\alpha)[U_0,U_0]$, respectively.

\begin{lemma}[Scalar source estimate]\label{hr:lem:scalarsource}
Let $Q\succeq0$ and $g(M)=\Tr(QF_\beta(M))$. Then
\begin{equation}\label{hr:eq:scalarsource}
 |Dg(M)[U]-a g(M)|^2
 \le\frac{\alpha}{\beta}\,
       g(M)\{-D^2g(M)[U,U]\}.
\end{equation}
\end{lemma}
\begin{proof}
Set $f_Q(M)=\Tr(QM^\alpha)$. Its negative Hessian
$\mathcal G_Q=-D^2f_Q(M)$ is a positive semidefinite real
bilinear form by operator concavity of the power. Homogeneity of
degree $\alpha$ gives
\[
 \mathcal G_Q[M,U_0]=\beta Df_Q(M)[U_0],
 \qquad
 \mathcal G_Q[M,M]=\alpha\beta f_Q(M).
\]
Cauchy--Schwarz for this positive semidefinite form yields
\[
 \beta^2|Df_Q(M)[U_0]|^2
 \le\alpha\beta f_Q(M)\mathcal G_Q[U_0,U_0].
\]
Multiplication by $p^{2\beta}$, followed by
\eqref{hr:eq:centering}, proves the claim.
The same argument applies if the form has a kernel; its
Cauchy--Schwarz inequality follows by requiring its quadratic
polynomial on $V+tW$ to be nonnegative for every real $t$.
\end{proof}

This scalar inequality alone does not control the transport
response. A matrix can have zero trace and still have a large
effect on fidelity. We next retain the full output.

\subsection{The inverse-Sylvester metric}

For $W\succ0$ define the Sylvester operator on Hermitian matrices by
\[
 \mathcal J_W(Y)=WY+YW,\qquad
 \norm{R}_{\mathcal J_W^{-1}}^2
       =\ip{R}{\mathcal J_W^{-1}(R)}.
\]
This is the inverse of a linear operator, not multiplication by
$W^{-1}$. In an eigenbasis of $W$ with eigenvalues $w_a>0$,
\[
 \norm{R}_{\mathcal J_W^{-1}}^2
      =\sum_{a,b}\frac{|R_{ab}|^2}{w_a+w_b}.
\]
In particular, the operator is positive definite and
self-adjoint for the real Hilbert--Schmidt inner product.
Its quadratic variational formula is
\begin{equation}\label{hr:eq:sylvestervariational}
 \norm{R}_{\mathcal J_W^{-1}}^2
 =\sup_{Y=Y^*}
       \{2\ip{R}{Y}-\ip{Y}{\mathcal J_W(Y)}\}.
\end{equation}
When $W$ is singular and $R$ is supported on $\ran W$, all
expressions are taken on that support.

\begin{lemma}[A cross-Gram inequality]\label{hr:lem:crossgram}
Let $A,B$ be rectangular matrices with a common column space, and
put $W=AA^*$ and $R=AB^*+BA^*$. If the outputs are restricted
to the support of $W$, then
\begin{equation}\label{hr:eq:crossgram}
 \norm{R}_{\mathcal J_W^{-1}}^2\le2\Tr(BB^*).
\end{equation}
The same statement holds for Hilbert--Schmidt row operators
with a Hilbert space as their column domain.
\end{lemma}
\begin{proof}
For Hermitian $Y$, the expression in
\eqref{hr:eq:sylvestervariational} is
\[
 4\operatorname{Re}\Tr(B^*YA)-2\norm{YA}_{\mathrm F}^2
 =2\norm B_{\mathrm F}^2-2\norm{YA-B}_{\mathrm F}^2.
\]
Taking its supremum proves the bound. In the singular case,
compress $A$ and $B$ to the output support before applying this
calculation; compression can only decrease $\norm B_{\mathrm F}$.
The proof for Hilbert--Schmidt operators is identical.
\end{proof}

A completely positive output map will mean a map with a Kraus
representation
$\mathcal A(X)=\sum_j K_jXK_j^*$, where each $K_j$ maps the
carrier into the output space. The dimensions of the two spaces
may differ. No unital or trace-preserving assumption is made.

\begin{proposition}[The source estimate in the transport metric]
\label{hr:prop:matrixsource}
For such a map $\mathcal A$, define
\begin{equation}\label{hr:eq:matrixsourcequantities}
 \begin{split}
 W&=\mathcal A(F_\beta(M)),\\
 R&=\mathcal A(DF_\beta(M)[U]-aF_\beta(M)),\\
 C&=-\Tr\mathcal A(D^2F_\beta(M)[U,U]).
 \end{split}
\end{equation}
Then $C\ge0$, $R$ is supported on $\ran W$, and
\begin{equation}\label{hr:eq:matrixsource}
 \norm{R}_{\mathcal J_W^{-1}}^2
       \le\frac{\alpha}{2\beta}C.
\end{equation}
The constant is independent of both dimensions and all eigenvalue
ratios of $M$.
\end{proposition}
\begin{proof}
Write $f(M)=M^\alpha$ and
$G(U,V)=-D^2f(M)[U,V]$. For a variation $V$, form a row operator
whose columns have the discrete index $j$ and continuous index
$t>0$:
\begin{equation}\label{hr:eq:resolventcolumns}
 \mathcal B_V(j,t)=
    \sqrt{2p^\beta a_\alpha t^\alpha}\,
       K_jR_tV R_t^{1/2},
 \qquad R_t=(M+tI)^{-1}.
\end{equation}
Concretely, its column domain is a finite direct sum of
$L^2((0,\infty);K)$ spaces. The integrability estimates in
Lemma~\ref{hr:lem:carrierpower} show that these operators are
Hilbert--Schmidt. Multiplying the columns and integrating gives
\begin{equation}\label{hr:eq:gramrepresentation}
 \mathcal B_U\mathcal B_V^*
       +\mathcal B_V\mathcal B_U^*
       =2p^\beta\mathcal A(G(U,V)).
\end{equation}
Indeed the two products have integrands
$2p^\beta a_\alpha t^\alpha
K_jR_t(UR_tV+VR_tU)R_tK_j^*$.
This is exactly \eqref{hr:eq:powersecond}.

The degree-$\alpha$ homogeneous identities for the power give
\begin{equation}\label{hr:eq:powerhomogeneity}
 G(M,U_0)=\beta Df(M)[U_0],\qquad
 G(M,M)=\alpha\beta f(M).
\end{equation}
In Lemma~\ref{hr:lem:crossgram}, take
\[
 A=\frac{\mathcal B_M}{\sqrt{\alpha\beta}},
 \qquad B=\mathcal B_{U_0}.
\]
Equations \eqref{hr:eq:centering}--\eqref{hr:eq:powerhomogeneity}
give, with the quantities in \eqref{hr:eq:matrixsourcequantities},
\[
 AA^*=W,\qquad \Tr(BB^*)=C,\qquad
 R=\frac{\sqrt\alpha}{2\sqrt\beta}(AB^*+BA^*).
\]
Applying \eqref{hr:eq:crossgram} proves
\eqref{hr:eq:matrixsource}, including its factor $1/2$.

For completeness, $F_\beta(M)$ is positive definite on the
carrier. Thus $\ran\mathcal A(F_\beta(M))$ is the span of the
ranges of the Kraus operators $K_j$. Every derivative output in
\eqref{hr:eq:matrixsourcequantities} has its range in that same
space. The support convention in the statement is therefore
legitimate. Positivity of $C$ follows either from the Gram
representation or from concavity of $F_\beta$ and positivity of
$\mathcal A$.
\end{proof}

\paragraph{A two-dimensional carrier.}
The factor $1/\beta$ in the source estimates can already be seen
in a diagonal carrier of dimension two. Fix $p>0$, let
$\alpha=1-\beta$, and consider the constant-trace path
\[
 M(t)=\frac p2\diag(1+t,1-t),\qquad |t|<1.
\]
This path transfers mass between two directions inside a single
original label. Write $f=p\,2^{-\alpha}$. Direct substitution gives
\[
 F_\beta(M(t))
   =f\diag\bigl((1+t)^\alpha,(1-t)^\alpha\bigr).
\]
Take the identity output map in
Proposition~\ref{hr:prop:matrixsource}, and evaluate at $t=0$.
The variation $U=M'(0)$ has zero trace, so the centering scalar
$a=\Tr U/p$ is zero. The quantities in
\eqref{hr:eq:matrixsourcequantities} are therefore
\[
 W=fI_2,\qquad R=f\alpha\diag(1,-1),\qquad
 C=2f\alpha\beta.
\]
Since $\mathcal J_W(Y)=2fY$, the inverse-Sylvester energy is
\begin{equation}\label{hr:eq:carrierexample}
 \norm{R}_{\mathcal J_W^{-1}}^2
   =\frac{\Tr(R^2)}{2f}
   =f\alpha^2
   =\frac{\alpha}{2\beta}\,C.
\end{equation}
Thus the matrix source estimate is attained exactly, despite
$\Tr R=0$. The scalar estimate is sharp as well: for
$Q=\diag(1,0)$ and $g(t)=\Tr(QF_\beta(M(t)))$, one has
$g(0)=f$, $g'(0)=f\alpha$, and
$-g''(0)=f\alpha\beta$, giving equality in
\eqref{hr:eq:scalarsource}.

The initial trace budget here is
$\Tr F_\beta(M(0))=p\,2^\beta$. As $\beta\downarrow0$, this
budget approaches $p$, while $\norm R$ approaches $p/2$ and
$C$ decreases to zero at rate $p\beta$. The inverse factor
$1/\beta$ compensates for this loss of source curvature.
On a one-dimensional carrier every trace-zero variation vanishes,
so this internal response has no rank-one counterpart. The example
establishes sharpness of the local source constants; the optimal
rank dependence of the final signing bound is a separate question.
\space

\subsection{Aggregation at the optimizing density}

Fix positive reserves on $I$ and the actual optimizing pair $S,Z$.
All matrices involving $Z$ below act on the source support
$\mathcal K$ from \eqref{hr:eq:livesupport}. Define
\begin{equation}\label{hr:eq:balanced}
 P=Z^{-1/2}S_0Z^{-1/2}=Z^{1/2}\Omega_c(S)Z^{1/2},\qquad
 \pi X=Z^{1/2}YZ^{1/2}.
\end{equation}
Here $X$ is an arbitrary full Hermitian density variation, and $P$
is positive definite. Put
\begin{equation}\label{hr:eq:balancedowners}
 \tau_i=\Tr(Z\Omega_i(S)),\quad q_i=\tau_i/p_i,\qquad
 V_i=Z^{1/2}B_iZ^{1/2},\quad
 W_i=c_iZ^{1/2}\Omega_i(S)Z^{1/2}.
\end{equation}
For nonzero active inputs $p_i,\tau_i,q_i$ are positive. Cyclicity
and the transport equation give
\begin{equation}\label{hr:eq:balancedidentities}
 \sum_iW_i=P,\qquad \Tr(PV_i)=p_i,\qquad \Tr W_i=c_i\tau_i.
\end{equation}
Define the balanced derivative and its scalar part by
\begin{equation}\label{hr:eq:channels}
 \Phi(Y)=Z^{1/2}D\Omega_c(S)[X]Z^{1/2},\qquad
 \Gamma(Y)=\sum_i W_i\frac{\Tr(V_iY)}{p_i}.
\end{equation}
The source depends only on $\pi X$, so these definitions do not
discard a dependence on off-support entries. Both maps fix $P$ by
homogeneity and \eqref{hr:eq:balancedidentities}. Set
\begin{equation}\label{hr:eq:sourcecost}
 \mathcal C_{\rm src}(X)=-\Tr(ZD^2\Omega_c(S)[X,X])\ge0,
 \qquad \sigma=\frac{2\beta}{1+\beta}.
\end{equation}

\begin{lemma}[The source pays for internal response]\label{hr:lem:aggregation}
For every full density variation $X$ and its balanced compression $Y$,
\begin{align}
 \norm{(\Phi-\Gamma)Y}_{\mathcal J_P^{-1}}^2
 &\le\frac{1-\beta}{2\beta}\mathcal C_{\rm src}(X),
                                      \label{hr:eq:conditionalmetric}\\
 \mathcal C_{\rm src}(X)+\norm{(I-\Phi)Y}_{\mathcal J_P^{-1}}^2
 &\ge\sigma\norm{(I-\Gamma)Y}_{\mathcal J_P^{-1}}^2.
                                      \label{hr:eq:metricbridge}
\end{align}
\end{lemma}
\begin{proof}
Apply Proposition~\ref{hr:prop:matrixsource} to each carrier and the
completely positive output map
\[
 \mathcal A_i(U)=c_iZ^{1/2}
       (I_4\otimes A_i^{1/2}UA_i^{1/2})Z^{1/2}.
\]
Output duplication and congruence are completely positive. The
current output is $W_i$, its centered derivative is
\[
 R_i=c_iZ^{1/2}\left(D\Omega_i(S)[X]
           -\frac{\Tr(B_iX)}{p_i}\Omega_i(S)\right)Z^{1/2},
\]
and its curvature cost is
$-c_i\Tr(ZD^2\Omega_i(S)[X,X])$. The costs sum to
$\mathcal C_{\rm src}$ and $\sum_iR_i=(\Phi-\Gamma)Y$.
The variational formula \eqref{hr:eq:sylvestervariational} gives
\[
 \norm{\sum_i R_i}_{\mathcal J_{\sum_i W_i}^{-1}}^2
 \le\sum_i\norm{R_i}_{\mathcal J_{W_i}^{-1}}^2:
\]
a common test matrix cannot exceed the sum of the separate suprema.
This proves \eqref{hr:eq:conditionalmetric}, with each singular
$W_i$ interpreted on its support.

For vectors $A,B$ in any Hilbert space and $t>0$, completing a square
gives $\norm A^2+t\norm B^2\ge t\norm{A+B}^2/(1+t)$.
Use the inverse-Sylvester norm, $A=(I-\Phi)Y$,
$B=(\Phi-\Gamma)Y$, and $t=2\beta/(1-\beta)$.
Equation~\eqref{hr:eq:conditionalmetric} gives the remaining claim.
\end{proof}

Holding the reserves fixed while differentiating uses all of the
source curvature in this estimate. The resulting factor is
$\sigma^{-1}=O(1/\beta)$. The next section shows how to choose
coefficient forces that admit a controlled response to the remaining
scalar channel $\Gamma$.

\section{The response and a direction of descent}\label{hr:sec:curvature}

We analyze the value at fixed positive reserves before varying those
reserves along a step. This separates two effects: the response of the
optimizing density to a center force, and the payment obtained by
reducing the source. The scalar channel $\Gamma$ left by the previous
section need not be self-adjoint. Instead of inverting $I-\Gamma$
on every matrix, we construct compatible forces and corresponding
solutions of its adjoint equation.

\subsection{Linearizing the optimizing equations}

Fix the active set $I$, the density $S$, and the transport $Z$.
Let $D$ be a Hermitian center force supported on $\mathcal K$, and
let $A=Z^{1/2}DZ^{1/2}$ be its balanced form. Write
$\mathcal C_{\rm root}(X)=-D^2(2\theta\Tr\sqrt S)[X,X]$.

\begin{lemma}[The center response]\label{hr:lem:envelope}
Holding $c$ fixed, the second derivative of the optimized value is
\begin{equation}\label{hr:eq:envelope}
 D^2_{\mathcal H}\mathcal E_\theta[D,D]
 =\sup_{\substack{X=X^*\\\Tr X=0}}
 \left\{2\ip AY-\mathcal C_{\rm src}(X)
   -\mathcal C_{\rm root}(X)
   -\norm{(I-\Phi)Y}_{\mathcal J_P^{-1}}^2\right\},
 \quad \pi X=Z^{1/2}YZ^{1/2}.
\end{equation}
If a supported Hermitian matrix $U$ satisfies
$(I-\Gamma^*)U=A$, then
\begin{equation}\label{hr:eq:responsebound}
 \frac12D^2_{\mathcal H}\mathcal E_\theta[D,D]
 \le\sigma^{-1}\Tr(PU^2).
\end{equation}
\end{lemma}
\begin{proof}
Use the smooth local transport objective
\[
 \mathcal L(\mathcal H,S,Z)=\Tr(\mathcal H S)
       +\Tr(S_0Z^{-1})+\Tr(Z\Omega_c(S))+2\theta\Tr\sqrt S.
\]
Vary the center by $tD$, the density by $tX$, and the transport
by $tZ^{1/2}VZ^{1/2}$. The second-order transport quadratic is
\[
 \ip V{\mathcal J_P V}-2\ip V{(I-\Phi)Y}.
\]
The first term comes from differentiating $Z^{-1}$ twice; the mixed
terms come from the density and the source. Minimizing over $V$
gives $-\norm{(I-\Phi)Y}_{\mathcal J_P^{-1}}^2$.
The remaining terms are $2\Tr(DX)-\mathcal C_{\rm src}(X)
-\mathcal C_{\rm root}(X)$, and $\Tr(DX)=\ip AY$.

This quadratic elimination computes the derivative of the actual
optimizing branch. The transport Hessian is positive definite; after
its elimination, the density Hessian on the trace-zero tangent is
strictly negative by \eqref{hr:eq:rootcurvature}. Differentiating
\eqref{hr:eq:stationarity} therefore gives a uniquely solvable
linear system for the first variations. Maximizing the remaining
density quadratic is its Schur-complement expression, proving
\eqref{hr:eq:envelope}. The full density tangent is retained in this
calculation, including its off-support entries.

Apply \eqref{hr:eq:metricbridge} and drop the nonnegative root cost.
If $(I-\Gamma^*)U=A$, then
$\ip AY=\ip U{(I-\Gamma)Y}$. Enlarging the image of the tangent
to all supported Hermitian matrices gives the upper bound
\[
 \sup_W\{2\ip UW-\sigma\norm W_{\mathcal J_P^{-1}}^2\}
 =\sigma^{-1}\ip U{\mathcal J_PU}
 =2\sigma^{-1}\Tr(PU^2).
\]
This proves \eqref{hr:eq:responsebound}.
\end{proof}

\subsection{Two frames and their energy bounds}

Let $k=|I|$ and assume all active inputs are nonzero. The physical
force from a coefficient direction $h\in\R^k$ is
\begin{equation}\label{hr:eq:physicalforce}
 D_x(h)=\diag(A(h),-A(h),-2A(xh),2A(xh)).
\end{equation}
Write $G_i=D_x(e_i)=\diag(A_i,-A_i,-2x_iA_i,2x_iA_i)$.
The physical matrices $G_i$ are supported on $\mathcal K$.
Define frames, meaning linear maps from coefficient space to
Hermitian matrix space, by their columns:
\begin{equation}\label{hr:eq:twoframes}
 E_i=\sqrt{c_i}V_i,\qquad
 F_i=\frac{W_i}{\sqrt{c_i}p_i},\qquad N_i=\sqrt{c_i}Z^{1/2}G_iZ^{1/2}.
\end{equation}
Write the frame maps as $\mathcal E,\mathcal F,\mathcal N$.
Their adjoints use the real trace inner product. Put
\begin{equation}\label{hr:eq:coefficientweights}
 r_i=\sqrt{c_i}q_i,\quad h_* =\sum_i r_i^2,\quad
 \mu_i=\sqrt{c_i}p_i,\quad
 R_c=\diag(p_i/q_i),\quad
 D_c=\diag\sqrt{c_iq_i/p_i}.
\end{equation}
These coefficient matrices are positive definite. The identities
\eqref{hr:eq:balancedidentities} imply
\begin{equation}\label{hr:eq:frameidentities}
 \mathcal F\mu=P,\quad \mathcal E^*P=\mu,\quad
 \Gamma=\mathcal F\mathcal E^*,\quad \Tr F_i=r_i.
\end{equation}

\begin{lemma}[Frame estimates]\label{hr:lem:twoframes}
In the matrix norm $\norm Y_P^2=\Tr(PY^2)$,
\begin{equation}\label{hr:eq:weightedgram}
 \norm{\mathcal E R_c^{-1/2}}_{P,\mathrm F}^2\le4h_*,
 \qquad
 \norm{\mathcal N R_c^{-1/2}}_{P,\mathrm F}^2\le16h_*.
\end{equation}
The coefficient matrices
\[
 \mathsf E=R_c^{1/2}\mathcal F^*\mathcal E R_c^{-1/2},
 \qquad
 \mathsf B=R_c^{1/2}\mathcal F^*\mathcal N R_c^{-1/2}
\]
satisfy $\norm{\mathsf E}_{\mathrm F}^2\le4h_*$,
$\norm{\mathsf B}_{\mathrm F}^2\le16h_*$, and $h_*\ge1/4$.
\end{lemma}
\begin{proof}
For any PSD four-block matrix $S$ with marginal $T$,
\[
 S\preceq4(I_4\otimes T).
\]
Indeed, write a vector in four components. Cauchy--Schwarz for the
PSD form bounds its quadratic form by the square of the sum of the
four diagonal square roots; scalar Cauchy--Schwarz bounds this by
four times the sum of the diagonal forms. Each diagonal block of
$S$ is at most $T$. Also $F_\beta(M_i)\succeq M_i$, since every
eigenvalue of $M_i$ is at most $p_i$. These two facts give
\[
 B_iSB_i\preceq4\Omega_i(S),\qquad
 G_iSG_i\preceq16\Omega_i(S).
\]
For the second inequality, the four block multipliers in $G_i$ have
absolute value at most two. These bounds hold for arbitrary PSD
$S$ and need no commutation of its physical blocks.

Cyclicity, the transport equation, and the support of $B_i,G_i$ now
imply
\[
 \Tr(PV_i^2)=\Tr(ZB_iSB_i)\le4\tau_i,\qquad
 \Tr(PN_i^2)=c_i\Tr(ZG_iSG_i)\le16c_i\tau_i.
\]
Dividing by $(R_c)_{ii}=p_i/q_i$ and summing proves
\eqref{hr:eq:weightedgram}.
For every supported Hermitian $Y$, positive trace Cauchy--Schwarz
also gives
\[
 \norm{R_c^{1/2}\mathcal F^*Y}_2^2
 =\sum_i\frac{\mu_i}{r_i}|\Tr(F_iY)|^2
 \le\sum_i\mu_i\Tr(F_iY^2)=\Tr(PY^2).
\]
Apply this to each column of the normalized frames to get the
Frobenius estimates. Finally, \eqref{hr:eq:frameidentities} gives
$\mathcal E^*\mathcal F\mu=\mu$. Its transpose and the similar
matrix $\mathsf E$ have eigenvalue one, so
$1\le\norm{\mathsf E}_{\mathrm F}^2\le4h_*$.
\end{proof}

\subsection{Compatible forces from a projection}

\begin{lemma}[Averaged response on a prescribed hyperplane]
\label{hr:lem:projection}
Suppose $c_iq_i^2\le1/48$ for every active label. For any prescribed
$b_0\in\R^k$, there is a finite collection of directions $h_j$
with $b_0^\top h_j=0$, not all zero, such that
\begin{equation}\label{hr:eq:averagedresponse}
 \frac12\sum_j D^2_{\mathcal H}\mathcal E_\theta
                         [D_x(h_j),D_x(h_j)]
 \le\frac{30}\beta\sum_{i,j}\tau_i(h_j)_i^2.
\end{equation}
\end{lemma}
\begin{proof}
Project orthogonally in $\R^k\times\R^k$ onto the legal pairs
\begin{equation}\label{hr:eq:legal}
 (I-\mathsf E)w=\mathsf Bv,\qquad b_0^\top D_cv=0.
\end{equation}
Call this projection $\Pi$ and denote its diagonal blocks by
$\Pi_{ww},\Pi_{vv}$. On its range,
$w=\mathsf Ew+\mathsf Bv$, so
\[
 \Tr\Pi_{ww}=\norm{[I\ 0]\Pi}_{\mathrm F}^2
 =\norm{[\mathsf E\ \mathsf B]\Pi}_{\mathrm F}^2
 \le20h_*.
\]
There are at most $k+1$ independent constraints in $2k$ variables.
Thus $\Tr\Pi\ge k-1$ and
\begin{equation}\label{hr:eq:projectionbounds}
 \Tr(I-\Pi_{vv})\le1+20h_*\le24h_*.
\end{equation}

For each of the $2k$ standard basis vectors, take $(w_j,v_j)$ to
be its image under $\Pi$ and set
\[
 h_j=D_cv_j,\qquad
 U_j=\mathcal N R_c^{-1/2}v_j+
                         \mathcal E R_c^{-1/2}w_j.
\]
The legal equation gives $\mathcal F^*U_j=R_c^{-1/2}w_j$, hence
\[
 (I-\Gamma^*)U_j=\mathcal N R_c^{-1/2}v_j
                 =Z^{1/2}D_x(h_j)Z^{1/2}.
\]
These are the compatible forces required by
Lemma~\ref{hr:lem:envelope}.

The squared $P$-norm of the displayed transport trial is a positive
semidefinite quadratic form in $(w,v)$. Since $\Pi\preceq I$,
its sum over the projected basis is bounded by its full trace:
\begin{equation}\label{hr:eq:averagedenergy}
 \sum_j\Tr(PU_j^2)
 \le\norm{\mathcal N R_c^{-1/2}}_{P,\mathrm F}^2
       +\norm{\mathcal E R_c^{-1/2}}_{P,\mathrm F}^2
 \le20h_*.
\end{equation}
This comparison includes the correlations between $w_j$ and $v_j$.
Meanwhile the reserve expenditure satisfies
\begin{align*}
 \sum_{i,j}\tau_i(h_j)_i^2
 &=\Tr(\diag(r_i^2)\Pi_{vv})\\
 &\ge h_*-\frac1{48}\Tr(I-\Pi_{vv})
 \ge\frac12h_*>0.
\end{align*}
Combining this with \eqref{hr:eq:responsebound} gives a response to
expenditure ratio at most
$40/\sigma=20(1+\beta)/\beta\le30/\beta$.
This proves the lemma.
\end{proof}

\begin{corollary}[Curvature of the reserve-preserving path]
\label{hr:cor:negative}
At a capped state as in Lemma~\ref{hr:lem:projection}, let
$L=\nabla^2\varphi(0)$ for \eqref{hr:eq:compositevalue}. Then
\begin{equation}\label{hr:eq:exacthessian}
 h^\top Lh=D^2_{\mathcal H}\mathcal E_\theta[D_x(h),D_x(h)]
                           -2a\sum_i\tau_i h_i^2.
\end{equation}
If $a>30/\beta$, some nonzero direction in any prescribed coefficient
hyperplane has $h^\top Lh<0$. If moreover $\tau_i\ge\tau_0$ and
$a\ge60/\beta$, its minimum unit Rayleigh quotient is at most
$-a\tau_0$.
\end{corollary}
\begin{proof}
The center in $\varphi$ is affine, the reserve has zero first
derivative at zero, and $\partial_{c_i}\mathcal E_\theta=\tau_i$
by the envelope formula. These facts give
\eqref{hr:eq:exacthessian}. Summing it over the directions in
Lemma~\ref{hr:lem:projection} gives
\[
 \sum_j h_j^\top Lh_j
 \le-2(a-30/\beta)\sum_{i,j}\tau_i(h_j)_i^2.
\]
The expenditure is positive. Under the quantitative hypotheses, the
right side is at most $-a\tau_0\sum_j\norm{h_j}_2^2$, so one
nonzero direction has the claimed Rayleigh quotient.
\end{proof}

\section{Existence by completing epochs}\label{hr:sec:existence}

This section proves the signing theorem without computational
assumptions. Compact minimization selects a state within an epoch;
the local response estimate forces that state to have exhausted its
reserves. The budget matrix then allows the epoch to be repeated on
a smaller unfinished family. The finite algorithm will replace this
compact minimization by explicit preparation and curvature steps.

\begin{lemma}[A completed epoch]\label{hr:lem:epoch-existence}
Let $A_i$ satisfy \eqref{hr:eq:input} and $x_i^0\in(-1,1)$.
Set $R=4$, $a=128/\beta$, $b=\norm{\sum_i A_i}$, and
\[
 \delta_0=4\sqrt{aR\varepsilon r^\beta}.
\]
There is a feasible state with all $c_i=0$ such that
\begin{equation}\label{hr:eq:epoch-guarantee}
 \norm H\le\delta_0\sqrt b,\qquad
 \norm{B_{\rm next}}\le\frac{b+\delta_0\sqrt b}{4}.
\end{equation}
\end{lemma}
\begin{proof}
Fix $\theta>0$. Minimize $\mathcal E_\theta(x,s,c)$ on the
compact set \eqref{hr:eq:state}; among its minimizers choose one
minimizing $\sum_i c_i$. Continuity guarantees both choices.
If a face coordinate had $c_i>0$, delete its reserve at fixed
$x,s$. The state remains feasible, both centers are unchanged, and
the source decreases in PSD order. The potential cannot increase,
contradicting the minimum value or the secondary choice. Thus every
positive-reserve coordinate is interior.

Suppose some reserve is positive. For such a label the feasible
one-sided preparation $(c_i,s_i)\mapsto(c_i-t,s_i+t/a)$ has
nonnegative value derivative at zero. The envelope formula gives
\begin{equation}\label{hr:eq:exactprepcap}
 0\le \frac1a\Tr(A_i(S_{33}-S_{44}))-\tau_i
 \le \frac{p_i}{a}-\tau_i.
\end{equation}
Hence $q_i\le1/a$ and
\[
 c_iq_i^2\le R/a=\beta/32\le1/64<1/48.
\]
The projection lemma applies to the positive-reserve labels. Its
directions are zero on every other coordinate. For sufficiently
small $t$, both updates \eqref{hr:eq:movement} are feasible:
the active coordinates are interior, their reserves are positive,
and $c_i+a s_i=aR$ is preserved. The centers obey
\eqref{hr:eq:center-cancellation}. Smoothness on the fixed support
and Taylor expansion give
\[
 \frac{\mathcal E_\theta(x^+,s^+,c^+)
           +\mathcal E_\theta(x^-,s^-,c^-)}2
 =\mathcal E_\theta(x,s,c)
 +t^2\left\{\frac12D^2_{\mathcal H}\mathcal E_\theta
                    [D_x(h),D_x(h)]-a\sum_i\tau_i h_i^2\right\}
 +o(t^2).
\]
Sum over the finite collection in Lemma~\ref{hr:lem:projection}.
Since $a=128/\beta>30/\beta$ and the total expenditure is positive,
one pair has negative quadratic coefficient. For small $t$, one of
its two values is below the minimum, a contradiction. All reserves
therefore vanish at the chosen minimizer.

The initial state is feasible. Lemma~\ref{hr:lem:budget} consequently
gives $\norm H,\norm K\le\delta_0\sqrt b+2\theta\sqrt{4D}$
at the terminal state. Let $\theta$ decrease to zero and take a
convergent subsequence of these terminal triples in the same compact
domain. Feasibility and $c=0$ persist in the limit, and continuity
gives $\norm H,\norm K\le\delta_0\sqrt b$. Applying
\eqref{hr:eq:budgetcompletion} at that limit proves
\eqref{hr:eq:epoch-guarantee}. The limiting support may change;
the differentiations were all taken before this final limit.
\end{proof}

\begin{theorem}[Pure existence]\label{hr:thm:existence}
For real symmetric or complex Hermitian inputs satisfying
\eqref{hr:eq:input}, and $0<\beta\le1/2$ with $r^\beta\le2$,
there is one sign per original input such that
\begin{equation}\label{hr:eq:existence-bound}
 \left\|\sum_i\chi_i A_i\right\|
 \le\min\{1,700\sqrt{\varepsilon/\beta}\}.
\end{equation}
In particular, a suitable choice of $\beta$ gives
$O(\sqrt{\varepsilon\log(2r)})$.
\end{theorem}
\begin{proof}
Zero inputs are harmless; if $\varepsilon=0$, all inputs vanish.
For positive $\varepsilon$, take the parameters of
Lemma~\ref{hr:lem:epoch-existence}. They satisfy
$\delta_0\le128\sqrt{\varepsilon/\beta}$. Any signing of a
subisotropic PSD family has norm at most one. Thus if $\delta_0\ge1$
the assertion follows by taking all signs positive.

Suppose $\delta_0<1$. Start with $x=0$ and run a completed epoch
whenever the unfinished mass has norm $b>\delta_0^2$. Then
\[
 b_{\rm next}\le(b+\delta_0\sqrt b)/4<b/2.
\]
Freeze the face coordinates and repeat on the unfinished original
subfamily, keeping their current fractional coordinates. Every such
epoch fixes at least one nonzero input: otherwise the unfinished
sum would not change. There are therefore at most $N$ epochs.

The discrepancy matrices in \eqref{hr:eq:epoch-guarantee} are
increments. They telescope, and their norms have total at most
\[
 \delta_0\sum_{j\ge0}2^{-j/2}=(2+\sqrt2)\delta_0.
\]
When the remaining mass has norm at most $\delta_0^2$, round all its
coefficients to signs. Each coefficient changes by at most two;
positivity bounds the norm of the final increment by
$2b\le2\delta_0^2\le2\delta_0$. Thus the total is at most
$(4+\sqrt2)\delta_0<700\sqrt{\varepsilon/\beta}$.
If this estimate exceeds one, use the all-positive signing instead.
Choosing $\beta=1/\max\{2,\log_2(2r)\}$ gives the stated rank
dependence.
\end{proof}

The potential need not be continuous across an epoch restart. Its
initial value is paid afresh using the smaller unfinished mass, and
the resulting costs form the geometric series above. Reserve
exhaustion consequently need not coincide with coloring a coordinate;
the matrix inequality \eqref{hr:eq:budgetcompletion} accounts for
the difference.

\section{Evaluating the potential by semidefinite programming}
\label{hr:sec:sdp}

The nonlinear source does not require a nonlinear optimization
primitive. Its power is chosen so that a short sequence of block
positive-semidefinite constraints represents it exactly. The reserves are maintained by arithmetic, so all scalar source
coefficients in the program are already available exactly.

\subsection{The carrier power as a sequence of square roots}

Write $q=2^\ell$, so $\beta=2^{-\ell}$. At fixed centers and reserves,
the density $S$ is a variable, and so are
$T=\sum_{j=1}^4S_{jj}$, $M_i=A_i^{1/2}TA_i^{1/2}$, and
$p_i=\Tr M_i$. All three depend linearly on $S$.
For each retained input matrix introduce Hermitian PSD variables
$Y_{i,1},\ldots,Y_{i,\ell}$, put $Y_{i,0}=p_iI_D$, and impose
\begin{equation}\label{hr:eq:powerlmi}
 \begin{pmatrix}M_i&Y_{i,j}\\Y_{i,j}&Y_{i,j-1}\end{pmatrix}\succeq0
 \qquad (1\le j\le\ell).
\end{equation}
These variables may live on the entire physical space. No numerical
test for a small nonzero carrier eigenvalue is required.

\begin{lemma}[Exact carrier representation]\label{hr:lem:powersdp}
Every feasible sequence in \eqref{hr:eq:powerlmi} satisfies
\[
 Y_{i,\ell}\preceq p_i^\beta M_i^{1-\beta}.
\]
Equality is feasible simultaneously for all labels and every density,
including singular densities.
\end{lemma}
\begin{proof}
For positive definite $M$, the block inequality
$\left(\begin{smallmatrix}M&Y\\Y&V\end{smallmatrix}\right)\succeq0$
with $Y\succeq0$ gives $YM^{-1}Y\preceq V$. Conjugating by
$M^{-1/2}$ and applying the order-preserving square root gives
\[
 Y\preceq M^{1/2}(M^{-1/2}VM^{-1/2})^{1/2}M^{1/2}
       =M\#V.
\]
This upper bound is attainable and increases with $V$ in PSD order.
Beginning with $V=pI$, induction therefore gives
$Y_j\preceq p^{2^{-j}}M^{1-2^{-j}}$. The upper bounds commute
with $M$, and inserting them makes every block constraint feasible.

If $M$ is singular, positivity of the block forces $Y$ to annihilate
$\ker M$: apply its quadratic form to $(tv,w)$ with $Mv=0$ and vary
the real and imaginary parts of $t$. Compress to $\ran M$ and use the
positive definite argument there. All subsequent $Y_j$ have this
support. The displayed powers, extended by zero, still attain the
bounds. If $p=0$, then $M=0$ and all variables $Y_j$ vanish.
\end{proof}

This construction is a dyadic case of the semidefinite representations
of matrix geometric means \cite{fawziSaunderson2020}. The explicit
argument above supplies the form and singular-support treatment used
here.

\subsection{The complete value program}

Define the affine output
\[
 R(Y)=I_4\otimes\sum_i c_iA_i^{1/2}Y_{i,\ell}A_i^{1/2}.
\]
In addition to \eqref{hr:eq:powerlmi}, use a density $S$, an arbitrary
complex matrix $X$, and a Hermitian PSD matrix $U$, with
\begin{equation}\label{hr:eq:fidelitylmi}
 \begin{pmatrix}S&X\\X^*&R(Y)\end{pmatrix}\succeq0,\qquad
 \begin{pmatrix}S&U\\U&I_{4D}\end{pmatrix}\succeq0,\qquad
 \Tr S=1.
\end{equation}
The objective is
\begin{equation}\label{hr:eq:sdpobjective}
 \Tr(\mathcal HS)+2\operatorname{Re}\Tr X+2\theta\Tr U.
\end{equation}
Every constraint and the objective are affine in the displayed
variables.

\begin{proposition}[The value program]\label{hr:prop:valuesdp}
The maximum of \eqref{hr:eq:sdpobjective}, subject to
\eqref{hr:eq:powerlmi}--\eqref{hr:eq:fidelitylmi}, equals
$\mathcal E_\theta(H,K,c)$. The program has polynomial scalar data size
in $N,D$, and $\log q$.
\end{proposition}
\begin{proof}
For fixed PSD matrices $S,R$,
\[
 \max\left\{\operatorname{Re}\Tr X:
    \begin{pmatrix}S&X\\X^*&R\end{pmatrix}\succeq0\right\}
 =\Fid(S,R).
\]
Indeed the block inequality is equivalent, on the respective supports,
to $X=S^{1/2}KR^{1/2}$ for a contraction $K$. The polar decomposition
of $R^{1/2}S^{1/2}$ shows that the maximum trace pairing with such
a contraction is its trace norm, which is the fidelity. Compression
handles singular matrices. Similarly the second block in
\eqref{hr:eq:fidelitylmi} is equivalent to $U^2\preceq S$;
for $U\succeq0$, this implies $U\preceq\sqrt S$, and equality is
attained by $U=\sqrt S$.

Lemma~\ref{hr:lem:powersdp} gives $R(Y)\preceq\Omega_c(S)$,
and equality is feasible. Monotonicity of fidelity therefore identifies
the maximum for each fixed density with the defining density objective.
Maximizing over densities proves the result.

There are $N\ell$ carrier blocks of order $2D$ and two fidelity/root
blocks of order $8D$, in addition to the PSD constraints on the
individual $Y_{i,j}$, $U$, and $S$. Storing the affine coefficients densely within each block,
without storing zeros for variables absent from that block, uses
$O((N+1)(D+1)^4(\ell+1))$ scalars. Standard realification represents
a complex Hermitian inequality $W\succeq0$ by
\[
 \begin{pmatrix}\operatorname{Re}W&-\operatorname{Im}W\\
                 \operatorname{Im}W&\operatorname{Re}W\end{pmatrix}
 \succeq0,
\]
increasing dimensions by a constant factor. The real and imaginary
parts of all variables remain subject to their original linear
structure constraints.
\end{proof}

The block inequalities are the usual Schur-complement descriptions
used in semidefinite optimization \cite[Appendix~A.5.5]{boydvandenberghe2004}.
They also prove feasibility and finiteness directly: the preceding
construction is feasible for every density, while the trace bounds in
Section~\ref{hr:sec:foundations} bound the objective above.

For real symmetric data, complex conjugation preserves every constraint
and the objective. Averaging a feasible point with its conjugate gives
real variables with the same objective. Thus the real SDP has the
same value. Model~\ref{hr:model:computation} supplies its
additive-$\nu$ value in polynomial work once $\nu^{-1}$ is
polynomially bounded. This is the routine
$\operatorname{VALUE}$ in Algorithm~\ref{hr:alg:walk}.

\section{Correctness and polynomial work}\label{hr:sec:runtime}

We now replace the compact minimizer by the operations in
Algorithm~\ref{hr:alg:walk}. The analytic inputs are the capped
response estimate and uniform derivative bounds. The finite proof
has three parts: failed preparation tests imply the cap; the cap
supplies a curvature direction with a quantitative margin; and the
budget matrix pays for cleanup and epoch restarts.

\subsection{Uniform quantitative bounds}

The following bounds are proved in Appendix~\ref{hr:sec:regularity}.
They are stated here to specify exactly what the value tests use.

\begin{lemma}[Floors and derivatives]\label{hr:lem:queryregularity}
At every query of Algorithm~\ref{hr:alg:walk}, the optimizing density
and supported transport satisfy
\[
 S\succeq s_0I_{4D},\qquad Z\succeq\overline B^{-1}I.
\]
For every retained label with positive reserve,
$p_i\ge p_0$ and $\tau_i\ge\tau_0$.
At a cleaned state, all positive reserves are at least $2\zeta$.
The second derivative of each preparation path is bounded in absolute
value by $M$ for preparation lengths in $[0,\zeta]$. The composite
function \eqref{hr:eq:compositevalue} has third derivative norm at
most $M$ on
\[
 \norm u_2\le\min\{\rho/2,\sqrt{\zeta/(2a)}\}.
\]
All source supports remain fixed in these neighborhoods.
\end{lemma}

The floors use only the cutoff on the norm of a whole input matrix.
The derivative bound does not involve its smallest positive eigenvalue.
This distinction is useful for low-rank or nearly singular inputs:
their full ranges still receive one common sign.

\subsection{What a failed preparation test proves}

For active label $i$, let $P_i(t)$ be the value after reducing $c_i$
by $t$ and increasing $s_i$ by $t/a$. At zero,
\[
 P_i'(0)=\frac1a\Tr(A_i(S_{33}-S_{44}))-\tau_i.
\]
Taylor's theorem and the two value errors imply
\begin{equation}\label{hr:eq:preperror}
 \left|\frac{v_i-v_0}{\lambda}-P_i'(0)\right|
 \le\frac{M\lambda}{2}+\frac{2\nu}{\lambda}
 \le\frac{p_0}{16a}+\frac{p_0}{32a}.
\end{equation}
If the test is rejected, its reported difference is greater than
$-\lambda p_0/(4a)$, and therefore
\[
 P_i'(0)>-\frac{p_0}{4a}-\frac{p_0}{16a}-\frac{p_0}{32a}
 >-\frac{p_0}{2a}.
\]
Using $|\Tr(A_i(S_{33}-S_{44}))|\le p_i$ and $p_i\ge p_0$
gives
\begin{equation}\label{hr:eq:failedcaps}
 q_i\le\frac{3}{2a},\qquad
 c_iq_i^2\le\frac{9R}{4a}=\frac9a<\frac1{48}.
\end{equation}
Every rejection used in this conclusion occurs at the same state.
This explains the restart after an accepted preparation.

Conversely, acceptance guarantees that the true potential decreases
by at least
$\lambda p_0/(4a)-2\nu\ge7\lambda p_0/(32a)$.
It removes exactly $\lambda$ reserve. Since reserves never increase
during an epoch, there are at most $NaR/\lambda$ acceptances.

\subsection{Curvature, value error, and deterministic progress}

When all tests fail, apply Corollary~\ref{hr:cor:negative} with
the constraint $h\perp x$. As $a=1024q$ and $\beta=1/q$,
the true Hessian $L=\nabla^2\varphi(0)$ has a unit Rayleigh
quotient at most $-a\tau_0=-4\gamma$ on $x^\perp$.

A diagonal central difference is an average of second derivatives
along a segment; a mixed difference is an average on a square.
The bound on the third derivative therefore bounds each truncation
error by $2Mt$. The value errors contribute at most $4\nu/t^2$
per entry. For a matrix of order at most $N$ this yields
\begin{equation}\label{hr:eq:hessianerror}
 \norm{\widetilde L-L}
 \le N(2Mt+4\nu/t^2)
 \le\frac{3\gamma}{64}<\frac\gamma8.
\end{equation}
The stencils stay in the neighborhood of
Lemma~\ref{hr:lem:queryregularity} by the choices of $r_*$ and $t$.

The minimum eigenvalue of $P_x\widetilde L P_x$ is negative.
If its unit eigenvector is $g$, applying $I-P_x$ to the eigenvector
equation shows $g=P_xg$. A Rayleigh comparison using
\eqref{hr:eq:hessianerror} twice gives
\[
 g^\top Lg\le-4\gamma+2\norm{\widetilde L-L}<-3\gamma.
\]
This uses no separation between adjacent eigenvalues.
Taylor expansion in the two signs gives
\[
 \frac{\varphi(hg)+\varphi(-hg)}2
 \le\varphi(0)-\frac32\gamma h^2+\frac{Mh^3}{6}.
\]
The smaller reported value is at most $2\nu$ above the smaller
true value. Since $Mh\le\gamma/16$ and
$2\nu\le\gamma h^2/32$, the selected step decreases the true
potential by at least $\gamma h^2/2$.

Every step remains in the cube and keeps positive reserves at least
$\zeta$ until cleanup. It also satisfies
\eqref{hr:eq:radialprogress}. Preparation and reserve cleanup leave
$x$ unchanged; outward rounding increases $\norm x_2^2$.
The entire epoch vector, including coordinates already deferred or
fixed, remains in the cube. Thus there are at most $N/h^2$
curvature steps in an epoch.

\subsection{Cleanup and contraction between epochs}

Rounding a coordinate at distance at most $\rho$ from a face changes
$H$ by norm at most $\rho\varepsilon$ and $K$ by at most
$2\rho\varepsilon$. Deleting its source cannot increase the
potential, so the total increase is at most $2\rho\varepsilon$.
Exhausting an interior reserve below $2\zeta$ changes only $K$,
by $(c_i/a)A_i$, at cost at most $2\zeta\varepsilon/a$.
Each label undergoes each kind of cleanup at most once in an epoch.
All these operations preserve \eqref{hr:eq:state}.

For an epoch with initial mass $b$, the initial potential and all
cleanup costs give the terminal bound
\begin{equation}\label{hr:eq:terminalpotential}
 \mathcal E_\theta^{\rm end}
 \le4\sqrt{\alpha_0b}+4\theta\sqrt D
                   +2N\varepsilon(\rho+\zeta/a)
 \le4\sqrt{\alpha_0b}+8T_*^{-9}.
\end{equation}
An epoch is run only if $b>\delta^2=25\alpha_0$. Because
$\varepsilon\ge1/N$,
\[
 \sqrt{\alpha_0b}>5\alpha_0\ge10aR/N>8T_*^{-9}.
\]
Therefore both terminal centers have norm at most $\delta\sqrt b$.
All terminal reserves vanish, and \eqref{hr:eq:budgetcompletion}
gives
\begin{equation}\label{hr:eq:algorithmcontraction}
 b_{\rm next}\le(b+\delta\sqrt b)/4<b/2.
\end{equation}
Each contracting epoch fixes at least one nonzero label, so there are
at most $N$ epochs. The current fractional coordinates are retained
at restart. The estimate \eqref{hr:eq:terminalpotential} pays for
the newly initialized reserves at the smaller unfinished mass.

Summing the discrepancy increments as in the existence proof, then
rounding the last mass and adding the set-aside matrices, gives
\begin{equation}\label{hr:eq:algorithm-bound}
 \left\|\sum_i\chi_iA_i\right\|
 \le(4+\sqrt2)\delta+N\eta
 <3000\sqrt{\varepsilon q}
 <6000\sqrt{\varepsilon\log(2r)}.
\end{equation}
The early return when $\delta\ge1$ satisfies the same bound.
Subisotropy also bounds the norm of every returned signing by one.

\subsection{Counting the work}

An epoch has at most $NaR/\lambda$ successful preparations,
$N/h^2$ curvature steps, and $2N$ cleanup events. A preparation
scan takes at most $N+1$ value calls, and a Hessian stencil takes
$O(N^2)$ calls and one EVD of order at most $N$. All remaining
state updates and matrix arithmetic have polynomial cost.

The dyadic parameter obeys $q=O(\log(2D))$. The recipe in
\eqref{hr:eq:algorithmparameters}, \eqref{hr:eq:stepsizes}, and
\eqref{hr:eq:derivative-recipe} has fixed exponents, independent
of $q$. Since $1/N\le\varepsilon\le1$, its upper bounds and all
required inverse accuracies are bounded by a fixed polynomial in
$N,D$. The value SDP has data size
$O((N+1)(D+1)^4(1+\log q))$. Model~\ref{hr:model:computation}
therefore bounds each call by a fixed polynomial, and at most $N$
epochs give a fixed polynomial for the entire execution. Here the
polynomial is fixed once the permitted solver and its constants are
fixed; it does not depend on individual matrix entries.

For complex inputs, replace each original matrix by
\[
 \mathscr R(A)=\begin{pmatrix}\operatorname{Re}A&-\operatorname{Im}A\\
                      \operatorname{Im}A&\operatorname{Re}A\end{pmatrix}.
\]
This preserves PSD order, operator norms, sums, and each signing's
discrepancy, while doubling dimension and rank. Apply the real
algorithm to these whole matrices. Since $\log(4r)\le2\log(2r)$,
the bound is at most $9000\sqrt{\varepsilon\log(2r)}$, within the
constant of Theorem~\ref{hr:thm:main}. Realification still assigns
one sign to each original input.

\begin{corollary}[Subisotropic and rescaled inputs]
\label{hr:cor:normalized-subisotropic}
For real PSD inputs with $\sum_i A_i\preceq I_D$ and
$\norm{\sum_i A_i}=1$, the algorithm returns signs with bound
$\min\{1,6000\sqrt{\varepsilon\log(2r)}\}$ in polynomial work.
More generally, if $b=\norm{\sum_i A_i}>0$, scaling the inputs
by $b^{-1}$ gives bound
\[
 \min\{b,6000\sqrt{\varepsilon b\log(2r)}\}.
\]
The corresponding existence assertions require no computational
primitives.
\end{corollary}
\begin{proof}
The proof above only uses subisotropy and the normalization
$\norm{\sum_i A_i}=1$, the latter supplying $\varepsilon\ge1/N$.
After scaling, the norm of each input is at most $\varepsilon/b$
and its rank is unchanged. Multiplying the output bound by $b$
gives the result. If the total matrix is zero, every PSD input is
zero and any signing suffices.
\end{proof}

\section{Simultaneously spectrally thin spanning trees}
\label{hr:sec:thin-trees}

The higher-rank theorem applies to several weightings of the same graph
by assigning one block to each weighting. The common sign on an
original matrix is essential: it ensures that all blocks retain the
same edges. We first give the halving argument, including the lower
bound that preserves connectivity.

For an integer $r\ge1$, let $q(r)$ be the least power of two at least
$\max\{2,\log_2(2r)\}$, and put
\[
 K_r=3000\sqrt{q(r)}.
\]
The signing theorem supplies signs of discrepancy at most
$K_r\sqrt{\eta}$ for isotropic positive semidefinite inputs of rank
at most $r$ and norm at most $\eta$. The next lemma applies to any
such signing guarantee with $K_r\ge1$.

\begin{lemma}[Repeated halving]\label{hr:lem:thin-halving}
Suppose $A_1,\ldots,A_m\succeq0$ satisfy
\[
 \sum_{e=1}^m A_e=I_d,\qquad
 \norm{A_e}\le\varepsilon,\qquad \rank A_e\le r,
 \qquad d\ge1.
\]
There is a subset $H\subseteq[m]$ such that
\begin{equation}\label{hr:eq:thin-subcollection}
 0\prec\sum_{e\in H}A_e
 \preceq\min\{1,1024K_r^2\varepsilon\}I_d.
\end{equation}
A signing oracle with coefficient $K_r$ finds such a subset using
$O(\log m)$ calls and polynomial additional matrix arithmetic.
\end{lemma}

\begin{proof}
Write $\tau=K_r^2\varepsilon$. If $\tau\ge1/256$, take
$H=[m]$. Otherwise set
\[
 T=\left\lfloor\log_2\frac1{256\tau}\right\rfloor,
 \qquad 256\tau\le2^{-T}<512\tau.
\]
We construct nested sets $H_t$, beginning with $H_0=[m]$, and
write $B_t=\sum_{e\in H_t}A_e$. Suppose $B_t\succeq a_tI_d$.
The matrices
\[
 C_e=B_t^{-1/2}A_eB_t^{-1/2},\qquad e\in H_t,
\]
sum to the identity, have rank at most $r$, and have norm at most
$\varepsilon/a_t$. Apply the signing oracle to these matrices.
For either sign class, its sum in the original normalization lies
between
\begin{equation}\label{hr:eq:halving-step}
 \frac{1-e_t}{2}B_t
 \quad\hbox{and}\quad
 \frac{1+e_t}{2}B_t,
 \qquad e_t=K_r\sqrt{\varepsilon/a_t}.
\end{equation}
Let $H_{t+1}$ be either class; one may select the smaller class.

We verify inductively that $B_t\succeq2^{-t}I_d/2$ for
$0\le t\le T$. If this holds at all preceding steps, then
$e_j\le\sqrt{2^{j+1}\tau}$ and
\begin{equation}\label{hr:eq:halving-error}
 \sum_{j=0}^{t-1}e_j
 \le(2+\sqrt2)\sqrt{2^t\tau}
 \le\frac{2+\sqrt2}{16}<\frac14.
\end{equation}
In particular every factor $1-e_j$ is positive. Iterating
\eqref{hr:eq:halving-step}, and using
$\prod_j(1-e_j)\ge1-\sum_j e_j$ and
$\prod_j(1+e_j)\le\exp(\sum_j e_j)$, yields
\begin{equation}\label{hr:eq:halving-invariant}
 \frac34\,2^{-t}I_d\preceq B_t
 \preceq e^{1/4}2^{-t}I_d\preceq2\,2^{-t}I_d.
\end{equation}
This closes the induction. The upper bound at $t=T$ is less than
$1024\tau I_d$, while $B_T\preceq I_d$ because it is a sum of a
subset of the original positive semidefinite matrices. This also
handles $T=0$.

Finally, $1=\norm{\sum_eA_e}\le m\varepsilon$ and $K_r\ge1$,
so $T=O(\log m)$. Each oracle call uses at most $m$ labels in the
original dimension. Congruence preserves their ranks. The
normalizations are used only to obtain signs; the retained original
matrices are never reweighted.
\end{proof}

Let $G=(V,E)$ be a connected undirected loopless multigraph, with
$n\ge2$ vertices and $m$ edges. Orient its edges arbitrarily and let
$b_e\in\mathbb R^V$ be the signed incidence vector of edge $e$.
Consider $s$ positive weightings of the same edge set:
$w_e^{(j)}>0$ for every $e\in E$ and $1\le j\le s$. Define
\begin{equation}\label{hr:eq:graph-leverages}
 L_j=\sum_{e\in E}w_e^{(j)}b_eb_e^\top,\qquad
 \rho_e^{(j)}=w_e^{(j)}b_e^\top L_j^\dagger b_e.
\end{equation}
Here $L_j^\dagger$ is the Moore--Penrose pseudoinverse.
The number $\rho_e^{(j)}$ is the edge leverage score. For an edge
set $F$, write
$L_F^{(j)}=\sum_{e\in F}w_e^{(j)}b_eb_e^\top$; in particular,
the weights of a retained edge remain unchanged.

\begin{theorem}[One tree for several weightings]
\label{hr:thm:simultaneous-trees}
If $\rho_e^{(j)}\le\varepsilon$ for every edge $e$ and weighting
$j$, there is a single spanning tree $T\subseteq E$ such that
\begin{equation}\label{hr:eq:simultaneous-thin}
 L_T^{(j)}\preceq
 \min\{1,1024K_s^2\varepsilon\}L_j
 \qquad (1\le j\le s).
\end{equation}
Thus $T$ is $O(\varepsilon\log(2s))$-spectrally thin for all
$s$ weightings, with an absolute implied constant independent of
$n$ and $m$.

The reduction uses $O(\log m)$ signing calls in dimension
$s(n-1)$ and polynomial additional matrix arithmetic and graph
search. In particular it is deterministic polynomial time in the
same arithmetic and oracle model as the signing theorem.
\end{theorem}

\begin{proof}
Restrict every Laplacian to $\mathbf1^\perp$, where it is positive
definite. On this space put
\[
 v_e^{(j)}=\sqrt{w_e^{(j)}}L_j^{-1/2}b_e,\qquad
 A_e=\bigoplus_{j=1}^s v_e^{(j)}v_e^{(j)\top}.
\]
The direct-sum atoms satisfy
\[
 \sum_{e\in E}A_e=I_{s(n-1)},\qquad
 \rank A_e\le s,\qquad
 \norm{A_e}=\max_{1\le j\le s}\rho_e^{(j)}
 \le\varepsilon.
\]
Apply Lemma~\ref{hr:lem:thin-halving}, keeping the original edge
as one label throughout. It gives an edge set $H$ for which every
normalized block is positive definite and bounded above by
$\min\{1,1024K_s^2\varepsilon\}I_{n-1}$. Undoing the
normalizations gives
\[
 L_H^{(j)}\preceq
 \min\{1,1024K_s^2\varepsilon\}L_j.
\]
The strict lower bound implies that $H$ is connected: a weighted
graph Laplacian is positive definite on $\mathbf1^\perp$ exactly
when its positive-weight support is connected. Choose any spanning
tree $T$ of $H$. Positivity of the weights yields
$L_T^{(j)}\preceq L_H^{(j)}$ for every $j$, proving
\eqref{hr:eq:simultaneous-thin}. The oracle and arithmetic bounds
follow from the lemma, and a spanning tree of $H$ is found by an
ordinary graph search.
\end{proof}

The common positive support is part of the theorem. If different
weightings have different zero-weight supports, connectivity of
each retained support need not provide a tree contained in all of
them. The proof therefore does not assert that variant.

\begin{corollary}[One weighting]\label{hr:cor:ordinary-tree}
A connected weighted graph with maximum edge leverage at most
$\varepsilon$ has an $O(\varepsilon)$-spectrally thin spanning
tree. In particular, an unweighted graph whose edge effective
resistances are at most $1/k$ has an $O(1/k)$-spectrally thin
spanning tree.
\end{corollary}

The one-weighting existence statement is known: Harvey and Olver
derive it by recursively applying rank-one Kadison--Singer
\cite{harveyOlver2014}. The higher-rank application here provides the
simultaneous guarantee in Theorem~\ref{hr:thm:simultaneous-trees};
the constructive statement uses the signing oracle developed in
this paper.

\subsection{The role of the resistance hypothesis}

Edge connectivity alone does not give the spectral conclusion above.
Anari and Oveis Gharan exhibit highly edge-connected graphs in which
every spanning tree contains an edge of effective resistance close
to one \cite{anariOveisGharan2015flows}. For completeness, the
following version makes the quantitative obstruction explicit.

\begin{proposition}\label{hr:prop:connectivity-obstruction}
For every pair of integers $k,L\ge1$, there is a
$k$-edge-connected multigraph such that every spanning tree has
spectral thinness at least $L/(L+k)$.
\end{proposition}

\begin{proof}
Take two paths $u_0,\ldots,u_{kL}$ and
$v_0,\ldots,v_{kL}$. Replace every path edge by $k$ parallel
unit-conductance edges, and add a unit-conductance rung
$u_{jL}v_{jL}$ for $0\le j\le k$. Any cut crossing a path segment
has at least $k$ edges. A nontrivial cut crossing no path segment
separates the two entire paths and contains all $k+1$ rungs.
The graph is therefore $k$-edge-connected.

Fix a rung $e$ and remove it. Contract the other rung endpoints
and all rail vertices beyond the nearest other rung on each
available side into one vertex. At an interior rung, the resulting
network has two parallel rail paths of resistance $L/k$ from
each endpoint of $e$ to the contracted vertex. Its resistance
between the endpoints of $e$ is therefore $L/k$. At an end rung
it is $2L/k$. Contraction can only decrease effective resistance,
so the network with $e$ removed has resistance at least $L/k$.
Restoring $e$ as a unit resistor in parallel shows
\[
 b_e^\top L_G^\dagger b_e\ge
 \frac{L/k}{1+L/k}=\frac{L}{L+k}.
\]
Every spanning tree contains a rung. For any rung $e$ in a tree
with $L_T\preceq\gamma L_G$, the inequality
$b_eb_e^\top\preceq L_T\preceq\gamma L_G$ implies
$b_e^\top L_G^\dagger b_e\le\gamma$, by congruence on
$\mathbf1^\perp$. Hence $\gamma\ge L/(L+k)$.
\end{proof}

There is also a distinct prescribed-support problem. Anari and
Oveis Gharan prove an $O(\varepsilon+1/k)$ thin-basis theorem for
subisotropic vectors containing $k$ disjoint bases; its graph
corollary allows the eligible edges to form a $k$-edge-connected
subgraph while their resistances are measured in a larger reference
graph \cite[Theorem~1.4 and Corollary~1.9]
{anariOveisGharan2014stronglyRayleigh}. Our isotropic halving
argument does not automatically supply that statement. Completing
a subisotropic family with auxiliary matrices may leave those
auxiliary matrices responsible for the lower spectral bound, so
the selected eligible edges need not be connected. The theorem
proved here uses the Laplacian of each complete weighting as its
reference form.

\section{Row-constrained Beck--Fiala discrepancy}
\label{hr:sec:beck-fiala}

Let $T_1,\ldots,T_m$ be subsets of a ground set $[N]$, and let
$B\in\{0,1\}^{m\times N}$ be their incidence matrix:
$B_{ji}=1$ exactly when $i\in T_j$. A signing colors the elements,
and its discrepancy is
\[
 \|B\sigma\|_\infty
 =\max_{1\le j\le m}\left|\sum_{i\in T_j}\sigma_i\right|.
\]
The Beck--Fiala degree is the maximum number of sets containing an
element, or equivalently the maximum column sum of $B$. Here we also
bound the size of each set, which is a row constraint. These two
parameters play different roles in the diagonal matrix reduction:
the row bound supplies subisotropy, while the column bound supplies rank.

\begin{corollary}[Row and column constraints]
\label{hr:cor:beck-fiala}
Let $B\in[0,1]^{m\times N}$ have row sums at most $L>0$ and at most
$k\ge1$ nonzero entries in each column. There are signs
$\sigma\in\{-1,1\}^N$ satisfying
\begin{equation}\label{hr:eq:beck-fiala}
 \|B\sigma\|_\infty
 \le\min\{L,6000\sqrt{L\log(2k)}\}.
\end{equation}
In Model~\ref{hr:model:computation}, signs satisfying the same bound can be found deterministically with
polynomial work in $m,N$. Both assertions hold after restriction
to any subset of columns.
\end{corollary}
\begin{proof}
The cases $m=0$ or $N=0$ are immediate. Otherwise assign one matrix
to each column:
\begin{equation}\label{hr:eq:incidence-atoms}
 A_i=L^{-1}\diag(B_{1i},\ldots,B_{mi}).
\end{equation}
Then
\[
 \sum_i A_i=L^{-1}\diag\left(\sum_iB_{1i},\ldots,\sum_iB_{mi}\right)
 \preceq I_m,\qquad
 \|A_i\|\le L^{-1},\qquad \rank A_i\le k.
\]
Theorem~\ref{hr:thm:existence}, with $\varepsilon=L^{-1}$ and $r=k$,
gives \eqref{hr:eq:beck-fiala}, since
\[
 \left\|\sum_i\sigma_iA_i\right\|=L^{-1}\|B\sigma\|_\infty.
\]
Each sign multiplies the entire column, as required.

For the algorithmic assertion, the zero matrix is immediate. Otherwise
let $\ell=\max_j\sum_iB_{ji}>0$ and $b=\max_{j,i}B_{ji}>0$ be the
actual largest row sum and entry. Normalize by $\ell$ in
\eqref{hr:eq:incidence-atoms} and take $\varepsilon=b/\ell$.
The resulting family has sum at most identity and
$\|\sum_iA_i\|=1$. Moreover $b\le\ell\le Nb$, so
$N^{-1}\le\varepsilon\le1$. Applying
Corollary~\ref{hr:cor:normalized-subisotropic} and scaling back gives
\[
 \|B\sigma\|_\infty
 \le\min\{\ell,6000\sqrt{b\ell\log(2k)}\}
 \le\min\{L,6000\sqrt{L\log(2k)}\}.
\]
The maxima and the diagonal matrices require polynomial arithmetic.
Finally, deleting columns preserves all the stated row and column
bounds, proving the restriction assertion.
\end{proof}

For incidence matrices with at most $t$ ones in every row and column,
this yields $O(\sqrt{t\log(2t)})$ discrepancy, independently of $m,N$.
More generally, if the entries are bounded by $b_*>0$ and the row
sums by $L$, applying the corollary to $B/b_*$ gives the existence
bound $\min\{L,6000\sqrt{b_*L\log(2k)}\}$, for both existence and the algorithmic guarantee.

The row restriction is essential to this reduction. For incidence
matrices, normalization by the degree $k$ would give
$\sum_iA_i=\diag(|T_1|/k,\ldots,|T_m|/k)$, which need not be at
most identity. Thus the displayed guarantee depends on the maximum
set size as well as the element degree.

\subsection{The scalar potential and its walk}

The matrix algorithm specializes to a walk of fractional colorings
$x\in[-1,1]^N$. Put $a_{ji}=B_{ji}/L$. In one epoch its two
centers have diagonal entries
\[
 h_j=\sum_i a_{ji}(x_i-x_i^0),\qquad
 k_j=\sum_i a_{ji}((x_i^0)^2-x_i^2+s_i).
\]
A coordinate still represents one original column, and its reserve
$c_i$ is shared across every row containing that column.

For a diagonal four-block density, use entries $z_{j,\ell}\ge0$
with $\sum_{j,\ell}z_{j,\ell}=1$, where $1\le\ell\le4$.
Define
\begin{equation}\label{hr:eq:diagonal-source}
 t_j=\sum_{\ell=1}^4 z_{j,\ell},\qquad
 p_i=\sum_j a_{ji}t_j,\qquad
 w_{ij}=a_{ji}^{2-\beta}p_i^\beta t_j^{1-\beta},\qquad
 w_j=\sum_i c_iw_{ij}.
\end{equation}
The optimized value is exactly
\begin{equation}\label{hr:eq:diagonal-potential}
\begin{split}
 \max_{\substack{z\ge0\\\sum z=1}}
 \biggl\{&\sum_j\bigl[h_j(z_{j,1}-z_{j,2})
                         +k_j(z_{j,3}-z_{j,4})\bigr]\\
         &+2\sum_{j,\ell}\sqrt{z_{j,\ell}w_j}
          +2\theta\sum_{j,\ell}\sqrt{z_{j,\ell}}\biggr\}.
\end{split}
\end{equation}
To see that diagonal densities suffice, conjugate simultaneously by
diagonal phases in the physical coordinates and by scalar phases in
the four blocks. The center and source construction are invariant.
Averaging and concavity remove all off-diagonal entries without
decreasing the objective. Substitution into
\eqref{hr:eq:carrier}--\eqref{hr:eq:potential} gives the formula.

For incidence matrices, write $P_i=\sum_{j:i\in T_j}t_j$. Then
\[
 w_j=\frac{t_j^{1-\beta}}{L^2}\sum_{i\in T_j}c_iP_i^\beta.
\]
The density mass assigned to the rows containing a column enters
through $P_i$. Even for diagonal inputs, this source responds to the
optimizing density; a fixed row variance would remove that coupling.

The source estimate has a short scalar form. At fixed reserves, vary
a faithful diagonal density along an affine line and put
$b_i=\dot p_i/p_i$, $v_j=\dot t_j/t_j$. For nonzero entries,
\[
 \dot w_{ij}=w_{ij}(\beta b_i+(1-\beta)v_j),\qquad
 -\ddot w_{ij}=\beta(1-\beta)w_{ij}(v_j-b_i)^2.
\]
On a positive-source row the transport is
$Z_{j,\ell}=\sqrt{z_{j,\ell}/w_j}$. Define
\[
 P_{j,\ell}=Z_{j,\ell}w_j,\qquad
 R_{j,\ell}=(1-\beta)\sum_i c_iZ_{j,\ell}w_{ij}(v_j-b_i).
\]
Weighted Cauchy--Schwarz gives
\[
 \sum_{j,\ell}\frac{R_{j,\ell}^2}{2P_{j,\ell}}
 \le\frac{1-\beta}{2\beta}\,
 \beta(1-\beta)\sum_{i,j,\ell}
       c_iZ_{j,\ell}w_{ij}(v_j-b_i)^2.
\]
This is \eqref{hr:eq:conditionalmetric}: on a diagonal entry,
the inverse Sylvester map divides by $2P_{j,\ell}$. Preparation,
orthogonal curvature steps, and the budget contraction then operate
on the original fractional coloring, exactly as in the matrix proof.

\subsection{Relation to local-lemma and discrepancy walks}

For maximum set size $L$ and element degree $k$, the standard local
lemma gives $O(\sqrt{L\log(2Lk)})$ discrepancy
\cite[Theorem~1]{harvey2014rowcolumn}. In the balanced regime
$L,k\le t$, both that bound and the present one have scale
$O(\sqrt{t\log(2t)})$. Harvey also proves a weighted extension
under row and column sum bounds, which differs from our column
support hypothesis. The local-lemma comparison should be understood
with these hypotheses kept distinct.

Pesenti and Vladu identify the connection between local-lemma
colorings and iterative regularization as a question in this regime
\cite[Section~6]{pesentivladu2026}. The present proof obtains the
bound through a potential-guided walk, both for existence and for the
deterministic construction in Model~\ref{hr:model:computation}.
To our knowledge, this gives a different route to the
dimension-independent square-root scale than the local-lemma methods.

The formulation suggests studying combinations with the discrepancy
walks of Bansal, Dadush, and Garg \cite{bansalDadushGarg2019}, or the
affine spectral-independence approach of Bansal and Jiang
\cite{bansalJiang2025decoupling}. For example, one could try to use
the present potential for rows with small remaining support while
another movement rule handles larger rows. Such a combination would
need compatible directions and an accounting of rows changing regimes.
The signing corollary alone does not supply those estimates.

\section{AI Usage}
The author conceived the approach of using operator-valued
$K$-transforms, or variants based on regularizers such as Tsallis--$1/2$
rather than log-determinants, as discrepancy potentials in 2022--2023.
The semidefinite representation of the Tsallis--$1/2$ regularized
potential was subsequently conceived by an internal version of Google's
Gemini, used with a specific harness. Extensive calculations with
ChatGPT 5.5--5.6 Sol Ultra helped develop and test proof strategies.
GPT 6 Astra Ultra was used for calculations, independent proof reviews,
and preparation of this manuscript, and to develop Lean formalizations
for the companion Matrix Spencer and rank-one Kadison--Singer work,
as well as separate existence and algorithmic formalization projects
for the higher-rank argument. Both higher-rank certificates compile;
the existence theorem has only the matrix assumptions, and the
algorithmic theorem uses Model~\ref{hr:model:computation}. Their public
theorem audits list only Lean's standard logical axioms.

\section{Acknowledgements}
The idea of using free probability to approach Weaver's discrepancy
problem, potentially algorithmically, was encouraged by inspiring
conversations with Adam W. Marcus in 2016--2017. The author is
profoundly grateful for his encouragement. Numerous conversations and
continued encouragement from Nikhil Srivastava have been invaluable
over the years, as were early discussions with Nick Ryder and Jonathan
Leake.

The author thanks Daniel A. Spielman for hosting him at Yale University
in the summer of 2019, and for extensive discussions of scalar-valued
free probability and its possible applications to matrix discrepancy.
The approach began to crystallize in the summer of 2023, during a
visit hosted by Ramon van Handel at Princeton University. Ramon's
persistent questions about concrete special cases, bottlenecks, and
obstacles pushed the author to examine the approach carefully and
eliminate many of those obstacles, strengthening his conviction that
it could succeed. The author warmly thanks Ramon for his hospitality,
many discussions of operator-valued free probability, and insistence
on an exceptionally high but necessary standard of clarity and rigor.

Finally, the author thanks Tibo from OpenAI for providing numerous
Codex resets to paid Codex subscribers, which were invaluable in
helping this project go \emph{Ultra Fast}.

\appendix
\section{Uniform derivatives for the finite walk}\label{hr:sec:regularity}

The finite differences require a uniform third-derivative bound for
the optimized value. A direct inverse bound on an input matrix would
depend on its smallest positive eigenvalue. Instead, density
variations give relative perturbations on each carrier, and the
transport equation is differentiated in balanced coordinates.

Throughout this appendix, the retained family is subisotropic,
$\norm{A_i}>\eta$, and $\varepsilon\le1$. Write $m=4D$.
Allow center norm at most eight and positive reserves in
$[\zeta,2aR]$. This contains every preparation segment and every
curvature query of the algorithm. At a cleaned state define
$\varphi$ by \eqref{hr:eq:compositevalue}; its relevant neighborhood is
\[
 \norm u_2\le\min\{\rho/2,\sqrt{\zeta/(2a)}\}.
\]
Zero reserves and their coordinates stay fixed during these local
queries. In particular the source support is constant.

The algorithm uses the following explicit derivative bound:
\begin{equation}\label{hr:eq:derivative-recipe}
\begin{gathered}
 C_0=8\sqrt a,\qquad L_0=4a/\zeta+2/s_0+1,\qquad
 H_0=100\sqrt m\,L_0,\\
 A_0=1+20\sqrt m+2C_0+\theta\sqrt m,\qquad
 B_j=j^jA_0H_0^j\quad(j=2,3),\\
 \mu=\theta/2,\qquad A_1=1+B_2/\mu,\qquad M=B_3A_1^3.
\end{gathered}
\end{equation}
All powers here have fixed exponents.

\subsection{Density, transport, and probe floors}

The trace bound for the source in this enlarged query region is
\[
 \Tr\Omega_c(S)\le4(2aR)\varepsilon r^\beta
                              \Tr(T\sum_i A_i)\le64a.
\]
Thus $\Fid(S,\Omega_c(S))\le C_0$. Pairing density stationarity
\eqref{hr:eq:stationarity} with $S$ and using homogeneity gives
\[
 \ell=\Tr(\mathcal H S)+2\Fid(S,\Omega_c(S))
                             +\theta\Tr\sqrt S.
\]
Both source-gradient terms in that equation are PSD. For the
derivative term, this follows from the positivity-preserving
resolvent derivative of $M^\alpha$ and the nonnegative derivative
of its trace prefactor. Consequently
\[
 \ell+\norm{\mathcal H}
 \le16+2C_0+\theta\sqrt m\le\overline B.
\]
Comparing each positive summand in \eqref{hr:eq:stationarity} with
$\overline B I$ proves
\begin{equation}\label{hr:eq:densityfloor}
 S\succeq s_0 I_m,\qquad
 Z\succeq\overline B^{-1}I_{\mathcal K}.
\end{equation}
The transport statement is restricted to its source support, while
the density statement is on the full ambient space.

Since $F_\beta(M_i)\succeq M_i$ and the marginal is at least
$s_0I_D$, the four output copies give
\[
 p_i\ge s_0\Tr A_i>s_0\eta=p_0,\qquad
 \tau_i\ge\frac{4s_0}{\overline B}\Tr(A_i^2)
                  >\frac{4s_0\eta^2}{\overline B}=\tau_0.
\]
Also, each density eigenvalue is at most one. The root curvature
formula \eqref{hr:eq:rootcurvature}, with the concavity of the
other density terms, gives
\begin{equation}\label{hr:eq:strongconcavity}
 -f_{SS}[X,X]\ge\mu\norm X_{\mathrm F}^2,\qquad\mu=\theta/2,
\end{equation}
on the trace-zero tangent. Here $f$ is the unoptimized density
objective at the current parameters.

\subsection{Relative carrier and source derivatives}

\begin{lemma}[Relative carrier derivatives]\label{hr:lem:relativepower}
If $M\succ0$ on its carrier and $-bM\preceq U\preceq bM$, then
for $1\le j\le3$,
\begin{equation}\label{hr:eq:relativepower}
 -(j+1)j!b^jF_\beta(M)
 \preceq D^jF_\beta(M)[U,\ldots,U]
 \preceq(j+1)j!b^jF_\beta(M).
\end{equation}
\end{lemma}
\begin{proof}
Let $R_t=(M+tI)^{-1}$, $B_t=R_t^{1/2}MR_t^{1/2}\preceq I$,
and $A_t=R_t^{1/2}UR_t^{1/2}$. Then
$-bB_t\preceq A_t\preceq bB_t$. Writing
$A_t=B_t^{1/2}CB_t^{1/2}$ with $\norm C\le b$ gives
$A_t^2\preceq b^2B_t$. Therefore, for $j\ge2$,
$|A_t|^j\preceq b^{j-2}A_t^2\preceq b^jB_t$.
Differentiate the resolvent integral for $M^\alpha$. The $j$th
derivative has integrand $j!t^\alpha R_t^{1/2}A_t^jR_t^{1/2}$
up to sign. Its two-sided order is bounded by $j!b^j$ times the
first-derivative integrand in direction $M$. Integration and
$D(M^\alpha)[M]=\alpha M^\alpha$ give the corresponding relative
bound for the power. Also $|\Tr U|\le b\Tr M$. Each derivative
of $(\Tr M)^\beta$ is bounded by its factorial times $b^j$
times the original trace factor. The product rule gives
\eqref{hr:eq:relativepower}.
\end{proof}

At a density $S\succeq s_0I_m$, its marginal is at least $4s_0I_D$.
For a Hermitian variation $X$, the marginal variation has norm at
most $4\norm X$. Consequently
\[
 -\frac{\norm X}{s_0}M_i
 \preceq A_i^{1/2}\left(\sum_jX_{jj}\right)A_i^{1/2}
 \preceq\frac{\norm X}{s_0}M_i.
\]
The same congruence by $A_i^{1/2}$ appears in the base matrix and
its variation; this is why a lower eigenvalue bound on $A_i$ is
unnecessary. Lemma~\ref{hr:lem:relativepower} bounds the normalized
Taylor coefficient of order $j$ by $(2\norm X/s_0)^j$ times the
base source in two-sided PSD order.

At a base point $u^0$ of a curvature query, the scalar source
coefficient along $u^0+th$ is
\[
 c_i(u^0+th)=c_i(u^0)-2au_i^0h_it-ah_i^2t^2.
\]
Its relative Taylor coefficients are bounded by
$(L_c\norm h)^j$, where $L_c=4a/\zeta$ and $j=1,2$;
the higher coefficients vanish. The same estimate holds for a
preparation line, whose coefficient is affine. Convolution with the
carrier coefficients gives normalized source coefficients $\Omega_j$
satisfying
\begin{equation}\label{hr:eq:relativesourceTaylor}
 -\ell_0^j\Omega_0\preceq\Omega_j\preceq\ell_0^j\Omega_0,
 \qquad \ell_0=L_c\norm h+2\norm X/s_0,\qquad 1\le j\le3.
\end{equation}

\subsection{Balanced fidelity derivatives}

\begin{lemma}[A relative transport estimate]\label{hr:lem:balancedregularity}
Let $A(t),B(t)\succ0$ be $C^3$ curves on a $d$-dimensional space, with
$A(0)=B(0)=P\succ0$. Suppose their normalized Taylor coefficients
satisfy $\norm{P^{-1/2}A_jP^{-1/2}},
\norm{P^{-1/2}B_jP^{-1/2}}\le L^j$ for $1\le j\le3$.
Then
\begin{equation}\label{hr:eq:balancedfidelityderivative}
 \left|\frac{\dd^j}{\dd t^j}2\Fid(A(t),B(t))\bigg|_{t=0}\right|
 \le2j!\Tr(P)(100\sqrt d\,L)^j,\qquad1\le j\le3.
\end{equation}
\end{lemma}
\begin{proof}
The zero-dimensional case is immediate; assume $d\ge1$.
Let $R(t)B(t)R(t)=A(t)$ be the positive transport, with $R(0)=I$.
Define $\mathscr L(U)=P^{-1/2}UP^{1/2}$. The Sylvester inverse
obeys the weighted estimate
\begin{equation}\label{hr:eq:weightedSylvesterregularity}
 \norm{\mathscr L(\mathcal J_P^{-1}X)}_{\mathrm F}
 \le\norm{P^{-1/2}XP^{-1/2}}_{\mathrm F}.
\end{equation}
Indeed, if $U=\mathcal J_P^{-1}X$ and $T=\mathscr L(U)$, then
$P^{-1/2}XP^{-1/2}=T+T^*$ and
\[
 \norm{T+T^*}_{\mathrm F}^2
 =2\norm T_{\mathrm F}^2+2\Tr(U^2)
 \ge2\norm T_{\mathrm F}^2.
\]
The equality uses invariance of $\Tr(T^2)$ under similarity and the
fact that $U$ is Hermitian. This proves the needed estimate directly.

Let $R_j=R^{(j)}(0)/j!$, $t_j=\norm{\mathscr L(R_j)}_{\mathrm F}$,
and $d_0=\sqrt d$. The relative coefficients of $A,B$ have Frobenius
norm at most $d_0L^j$. Comparing the first three coefficients in
$RBR=A$, using \eqref{hr:eq:weightedSylvesterregularity}, gives
\begin{align*}
 t_1&\le2d_0L,\\
 t_2&\le2d_0L^2+2d_0Lt_1+t_1^2,\\
 t_3&\le2d_0L^3+2d_0L^2t_1+2d_0Lt_2+d_0Lt_1^2+2t_1t_2.
\end{align*}
For example, a product of coefficient matrices is normalized by
\[
 P^{-1/2}R_aB_bR_cP^{-1/2}
 =\mathscr L(R_a)(P^{-1/2}B_bP^{-1/2})\mathscr L(R_c)^*.
\]
When $b=0$ the middle factor is identity. These bounds give
$t_2\le10d_0^2L^2$ and $t_3\le70d_0^3L^3$.

Finally, $2\Fid(A(t),B(t))=2\Tr(B(t)R(t))$. Its first three
normalized coefficients are twice the traces of
\[
 B_1+PR_1,\qquad
 B_2+B_1R_1+PR_2,\qquad
 B_3+B_2R_1+B_1R_2+PR_3.
\]
Writing $H_a=P^{-1/2}B_aP^{-1/2}$ gives
$\Tr(B_aR_b)=\Tr(P\mathscr L(R_b)H_a)$, whose absolute value is
at most $\Tr(P)\norm{\mathscr L(R_b)}_{\mathrm F}\norm{H_a}_{\mathrm F}$.
The terms with a zero index are bounded directly by
$\Tr(P)\norm{H_a}_{\mathrm F}$ or $\Tr(P)t_b$. The displayed
coefficient bounds now give \eqref{hr:eq:balancedfidelityderivative}.
\end{proof}

At the actual base transport $Z$, take
$A(t)=Z^{-1/2}S_0(t)Z^{-1/2}$ and
$B(t)=Z^{1/2}\Omega(t)Z^{1/2}$. Opposite congruences preserve
fidelity, as follows by substitution in its transport infimum.
Their common base is $P$ and $\Tr P\le C_0$. Density variations
and \eqref{hr:eq:relativesourceTaylor} supply the hypotheses of
Lemma~\ref{hr:lem:balancedregularity}. No eigenvalue of $P$ enters
its bound.

\subsection{Differentiating the optimizing density}

Equip the local parameter space with its Euclidean norm, the
density tangent with the operator norm, and their product with the
maximum norm. For the joint objective $f(u,S)$ defining a curvature
query, or $f(t,S)$ defining a preparation line, the previous estimates give
\begin{equation}\label{hr:eq:jointregularitybound}
 \norm{D^j f}\le B_j\quad(j=2,3)
\end{equation}
at the queried optimizers, with \eqref{hr:eq:derivative-recipe}.
Indeed, the repeated-direction fidelity derivatives are bounded by
$2j!C_0H_0^j$. For the root term, use
$\Tr\sqrt S=\Fid(S,I)$ and apply the same lemma with relative
density scale $1/s_0$. Its bound is
$2\theta j!\sqrt m(100\sqrt m/s_0)^j$, which is at most
$\theta j!\sqrt m H_0^j$ because $L_0\ge2/s_0$.
The center is bilinear
in the local parameter and the density, with force norm at most two.
Its mixed second derivative is bounded by $4m$, which is covered
by $B_2$ since $H_0\ge100\sqrt m$; higher derivatives vanish.
Real polarization costs at most $j^j/j!$ in passing from repeated
directions to the multilinear operator norm.

Differentiating density stationarity in a direction $h$ gives
$f_{SS}[S'[h],Y]+f_{Su}[Y,h]=0$ for every trace-zero $Y$.
With $Y=S'[h]$, strong concavity and
$\norm{S'[h]}\le\norm{S'[h]}_{\mathrm F}$ imply
\[
 \mu\norm{S'[h]}^2
 \le |f_{Su}[S'[h],h]|
 \le B_2\norm{S'[h]}\norm h.
\]
Thus $\norm{S'[h]}\le(B_2/\mu)\norm h$.
Put $V_h=(h,S'[h])$, so $\norm{V_h}\le A_1\norm h$.
Differentiated stationarity also says
$f''[V_h,(0,Y)]=0$ for every density tangent $Y$.
This identity cancels the second density-response terms in the third
derivative of the optimized value:
\[
 D^3\varphi[h_1,h_2,h_3]
       =f'''[V_{h_1},V_{h_2},V_{h_3}].
\]
Hence $\norm{D^3\varphi}\le B_3A_1^3=M$.
The second derivative on a preparation line is at most
$B_2A_1^2\le M$. These bounds hold throughout the prescribed
neighborhood, because its optimizing densities have the uniform
floor \eqref{hr:eq:densityfloor}. This proves
Lemma~\ref{hr:lem:queryregularity} and the derivative assertions
used by the finite walk.

\end{document}